\documentclass[preprint,nofootinbib,amsmath,amssymb,amsfonts,aps,prd]{revtex4-1}

\usepackage{setspace}
\usepackage{graphicx}
\usepackage{subfigure}
\usepackage[bookmarks=true, pdfstartview={FitH}, bookmarksopen=true, bookmarksopenlevel=1, linktoc=page]{hyperref}
\usepackage{amsthm}
\usepackage{mathrsfs}

\DeclareMathOperator{\lie}{\pounds}
\DeclareMathOperator{\bepsilon}{\boldsymbol{\epsilon}}
\DeclareMathOperator{\bomega}{\boldsymbol{\omega}}
\DeclareMathOperator{\bL}{\boldsymbol{L}}
\DeclareMathOperator{\btheta}{\boldsymbol{\theta}}
\DeclareMathOperator{\bN}{\boldsymbol{N}}
\DeclareMathOperator{\bJ}{\boldsymbol{J}}
\DeclareMathOperator{\bE}{\boldsymbol{E}}
\DeclareMathOperator{\bC}{\boldsymbol{C}}
\DeclareMathOperator{\bQ}{\boldsymbol{Q}}
\DeclareMathOperator{\bS}{\boldsymbol{S}}
\DeclareMathOperator{\mE}{\mathcal E}
\DeclareMathOperator{\mV}{\mathcal V}
\DeclareMathOperator{\mK}{\mathcal K}
\newcommand{\be}{\begin{equation}}
\newcommand{\ee}{\end{equation}}
\newcommand{\ba}{\begin{aligned}}
\newcommand{\ea}{\end{aligned}}

\usepackage{amsthm}
\let\oldproofname=\proofname
\renewcommand{\proofname}{\rm\bf{\oldproofname}:}
\newtheorem{theorem}{Theorem}[section]
\newtheorem*{corollary}{Corollary}
\newtheorem{lemma}{Lemma}[section]

\begin{document}

\pagestyle{myheadings}

\title{Dynamic and Thermodynamic Stability of Relativistic, Perfect
  Fluid Stars}
\author{Stephen R. Green}
\email{sgreen04@uoguelph.ca; CITA National Fellow}
\affiliation{Department of Physics\\
  University of Guelph\\
  Guelph, Ontario N1G 2W1, Canada}
\author{Joshua S. Schiffrin}
\email{schiffrin@uchicago.edu}
\author{Robert M. Wald}
\email{rmwa@uchicago.edu}
\affiliation{Enrico Fermi Institute and Department of Physics \\
  The University of Chicago \\
  5640 S. Ellis Ave., Chicago, IL 60637, U.S.A.}
\date{\today}

\begin{abstract} 
\begin{spacing}{1.3} 
  We consider perfect fluid bodies (``stars'') in general relativity,
  with the local state of the fluid specified by its $4$-velocity,
  $u^a$, its ``particle number density,'' $n$, and its ``entropy per
  particle,'' $s$.  A star is said to be in {\em dynamic equilibrium}
  if it is a stationary, axisymmetric solution to the Einstein-fluid
  equations with circular flow. A star is said to be in {\em
    thermodynamic equilibrium} if it is in dynamic equilibrium and its
  total entropy, $S$, is an extremum for all variations of initial
  data that satisfy the Einstein constraint equations and have fixed
  total mass, $M$, particle number, $N$, and angular momentum, $J$. We
  prove that for a star in dynamic equilibrium, the necessary and
  sufficient condition for thermodynamic equilibrium is constancy of
  angular velocity, $\Omega$, redshifted temperature, $\widetilde{T}$,
  and redshifted chemical potential, $\widetilde{\mu}$. A star in
  dynamic equilibrium is said to be {\em linearly dynamically stable}
  if all physical, gauge invariant quantities associated with linear
  perturbations of the star remain bounded in time; it is said to be
  {\em mode stable} if there are no exponentially growing solutions
  that are not pure gauge. A star in thermodynamic equilibrium is said
  to be {\em linearly thermodynamically stable} if $\delta^2 S < 0$
  for all variations at fixed $M$, $N$, and $J$; equivalently, a star
  in thermodynamic equilibrium is linearly thermodynamically stable if
  \mbox{$\delta^2 M - \widetilde{T} \delta^2 S -\widetilde{\mu} \delta^2 N -
    \Omega \delta^2 J > 0$} for all variations that, to first order,
  satisfy $\delta M = \delta N = \delta J = 0$ (and, hence, $\delta
  S=0$). Friedman previously identified positivity of canonical
  energy, $\mathcal E$, as a criterion for dynamic stability and
  argued that all rotating stars are dynamically unstable to
  sufficiently non-axisymmetric perturbations (the CFS instability), so our main focus is on
  axisymmetric stability (although we develop our formalism and prove
  many results for non-axisymmetric perturbations as well).  We show
  that for a star in dynamic equilibrium, mode stability holds with
  respect to all axisymmetric perturbations if $\mathcal E$ is
  positive on a certain subspace, $\mathcal V$, of axisymmetric
  Lagrangian perturbations that, in particular, have vanishing
  Lagrangian change in angular momentum density.  Conversely, if
  $\mathcal E$ fails to be positive on $\mathcal V$, then there exist
  perturbations that cannot become asymptotically stationary at late
  times.  We further show that for a star in thermodynamic
  equilibrium, for all Lagrangian perturbations, we have $\mathcal E_r
  = \delta^2 M - \Omega \delta^2 J$, where $\mathcal E_r$ denotes the
  ``canonical energy in the rotating frame,'' so positivity of $\mE_r$
  for perturbations with $\delta J = 0$ is a necessary condition for
  thermodynamic stability.  For axisymmetric perturbations, we have
  $\mathcal E = \mathcal E_r$, so a necessary condition for
  thermodynamic stability with respect to axisymmetric perturbations
  is positivity of $\mathcal E$ on all perturbations with $\delta J =
  0$, not merely on the perturbations in $\mathcal V$.  Many of our
  results are in close parallel with the results of Hollands and Wald
  for the theory of black holes.
\end{spacing}
\end{abstract}

\maketitle

\tableofcontents

\newpage

\section{Introduction} \label{intro}

In a recent paper, Hollands and Wald \cite{HollandsWald} applied Lagrangian methods to analyze the linear stability of black holes. They showed that the necessary and sufficient condition for dynamical stability of a black hole with respect to axisymmetric perturbations is the positivity of canonical energy, $\mathcal E$, on the subspace of perturbations with vanishing first order change in mass, angular momentum, and area, $\delta M = \delta J = \delta A= 0$. They further showed that for axisymmetric perturbations, the canonical energy is given in terms of second order variations by 
\begin{equation}
\mathcal E = \delta^2 M - \frac{\kappa}{8 \pi} \delta^2 A - \Omega_H \delta^2 J \, ,
\end{equation}
thereby showing that for perturbations with $\delta M = \delta J = \delta A= 0$, dynamical stability of a black hole is equivalent to its thermodynamic stability (with $\kappa/2\pi$ identified with temperature and $A/4$ identified with entropy). 

It is of interest to know whether similar results hold for the dynamic and thermodynamic stability of ordinary perfect fluid stars in general relativity. This question is in some ways simpler than the corresponding question for black holes, since the black hole horizon was the source of many of the difficulties and subtleties in the analysis of \cite{HollandsWald}. On the other hand, a number of new difficulties and subtleties arise in the Lagrangian formulation of the Einstein-fluid equations, which, in some ways, make the analysis of the relationship between dynamic and thermodynamic stability for fluid stars more difficult than for black holes.

There have been many previous analyses of the stability of relativistic fluid stars. The most relevant for our work are the analyses of Friedman \cite{Friedman} and Lindblom and Hiscock \cite{LindblomHiscock}. Based upon earlier work by Chandrasekhar \cite{Chandra} and Friedman and Schutz \cite{FriedmanSchutz}, Friedman investigated the dynamic stability of relativistic stars with respect to perturbations that arise in the Lagrangian displacement framework. He showed that positivity of canonical energy defined with respect to the time translation Killing field provides a criterion for dynamic stability, and he used this criterion to show that all rotating perfect fluid stars are dynamically unstable to non-axisymmetric perturbations of sufficiently high angular quantum number $m$ (the CFS instability). Lindblom and Hiscock investigated the effects of viscosity and thermal conductivity on ``short length scale'' perturbations of rigidly rotating stars with constant redshifted temperature. Their investigation is therefore closely related to the thermodynamic stability of stars that are in thermal equilibrium. They argued that the appropriate condition for stability with respect to such dissipative processes is the positivity of a similar canonical energy, $\mE_r$, that is defined with respect to the Killing field to which the fluid $4$-velocity is proportional. Lindblom and Hiscock also argued that such dissipative processes would damp out the dynamic instability found by Friedman for perfect fluid stars.

In this paper, we will give a comprehensive, unified analysis of the dynamic stability and thermodynamic equilibrium and stability of fluid stars in general relativity. The derivation of some of our results will be restricted to fluid  perturbations that can be described within the Lagrangian displacement framework, but many of our results, including all of the results of section \ref{thermo}, will apply to completely general fluid perturbations. 

With regard to dynamic stability, our analysis will largely reproduce results of Friedman with respect to non-axisymmetric perturbations, but there will be a number of significant differences and extensions in the axisymmetric case. In particular, we shall see that in the axisymmetric case, we must impose a physical restriction on the perturbations to which the positivity of canonical energy criterion can be directly applied: Axisymmetric stability can be directly tested via positivity of canonical energy only on a subspace, $\mV$, of perturbations that, in particular, have vanishing Lagrangian change in the angular momentum per particle, $\Delta j$. Nevertheless, we prove that positivity of canonical energy on this restricted class of perturbations implies mode stability for all perturbations---including those that cannot be described in the Lagrangian displacement framework---so positivity of $\mE$ on $\mV$ is necessary and sufficient for dynamic stability with respect to axisymmetric perturbations.

Our analysis of thermodynamic stability of stars that are in thermodynamic equilibrium will depart significantly from Lindblom and Hiscock and others in that it will be based upon consideration of the total entropy of perfect fluid solutions rather than particular forms of the dynamical equations for a dissipative fluid. As explained further below, our approach should lead to results equivalent to those that would be obtained by considering dissipative fluids---provided that all forms of dissipation are included---but our approach is completely incapable of yielding any information concerning the growth rate of any thermodynamic instability. We will show that the criterion found by Lindblom and Hiscock for short length scale perturbations---namely positivity of canonical energy, $\mE_r$, defined with respect to the Killing field to which the fluid $4$-velocity is proportional---arises in a very simple and natural way as a necessary condition for thermodynamic stability: For a perfect fluid star to be thermodynamically stable, $\mE_r$ must be positive on the space of perturbations that can be described within the Lagrangian displacement framework and satisfy\footnote{\footnotesize The additional conditions $\delta N = \delta S = 0$ of our criterion given below for thermodynamic stability hold  automatically for fluid perturbations describable in the Lagrangian displacement framework.\medskip} $\delta J = 0$, where $J$ denotes the total (ADM) angular momentum. It is easily seen that all rotating stars are thermodynamically unstable\footnote{\label{lh}\footnotesize This does not contradict the conclusion of Lindblom and Hiscock that viscosity and thermal conductivity can damp the CFS instability and thereby stabilize a star in the sense of vastly increasing the timescale needed for a star to radiate all of its angular momentum; however, on general thermodynamic grounds, dissipative processes can never stabilize a dynamically unstable equilibrium configuration, since dissipative processes can only increase the region of phase space that is accessible to the system under dynamical evolution.\medskip} to suitable non-axisymmetric perturbations. For axisymmetric perturbations, we have $\mE_r = \mE$, so a necessary condition for thermodynamic stability to axisymmetric perturbations is positivity of $\mE$ on all perturbations with $\delta J = 0$ rather than merely on perturbations in $\mV$.

We now state our assumptions concerning the type of perfect fluid matter we consider, and we will then define precisely what we mean by dynamic and thermodynamic equilibrium and stability of fluid stars.

By a perfect fluid, we mean matter that has a stress energy tensor of the form
\be \label{perflu}
T_{ab}=(\rho+p)u_a u_b + p g_{ab}
\ee
with $u^a u_a = -1$.
We will be concerned in this paper with perfect fluids whose local state is characterized by a ``particle number density," $n$, and an ``entropy per particle,"
$s$. The energy density $\rho$ is taken to be a prescribed function, $\rho = \rho(n,s)$, of the variables $(n,s)$. 
The temperature, $T$, and chemical potential, $\mu$, of the fluid are defined by 
\be
T \equiv \frac{1}{n} \frac{\partial\rho}{\partial s} \, , \quad \quad \quad \mu \equiv \frac{\partial\rho}{\partial n} - Ts \, .
\ee
The pressure, $p$, of the fluid is then assumed to be given by
\be \label{gibbsduhem}
p = -\rho + \mu n + Tsn = n\frac{\partial\rho}{\partial n}-\rho \, ,
\ee
where the first equality corresponds to the integrated form of the Gibbs-Duhem relation. The above definitions of $T$ and $\mu$ imply that the
local first law of thermodynamics
\be \label{loc1st}
d \rho = T d(ns) + \mu dn 
\ee
holds as a mathematical identity.  The function $\rho(n,s)$---which characterizes all of the properties of the fluid---is assumed to be chosen so that for all allowed $(n,s)$ we have
\be
\rho \geq 0 \, , \quad \quad p \ge 0 \, , \quad \quad T > 0 \, ,  \quad \quad 
0 \leq c_s^2 \leq 1 \, ,
\label{Tpc}
\ee
where
\be
c_s^2 \equiv \left(\frac{d p}{d\rho}\right)_s \equiv \frac{\partial p/\partial n}{\partial \rho/\partial n}
=\frac{ n\, \partial^2 \rho/\partial n^2}{\partial \rho/\partial n} \, .
\ee
In eq. \eqref{localthermostab} below, we will give additional conditions on the functional form of $\rho(n,s)$ that must be satisfied if local thermodynamic stability of the fluid is to hold. 

The fluid equations of motion consist of conservation of stress energy,
\be \label{eq:conservationofstressenergy}
\nabla_a T^{ab}=0,
\ee
together with conservation of the particle number current,
\be \label{eq:conservationofn}
\nabla_a (nu^a)=0.
\ee
The $u^a$ component of \eqref{eq:conservationofstressenergy} (conservation of energy) together with \eqref{eq:conservationofn} imply conservation of entropy along worldlines,
\be \label{eq:conservationofs}
u^a\nabla_a s=0.
\ee
More generally, any two of \eqref{eq:conservationofn}, \eqref{eq:conservationofs}, and the $u^a$ component of \eqref{eq:conservationofstressenergy} imply the third. 

By a perfect fluid star, we mean a globally hyperbolic, asymptotically flat solution of the Einstein-fluid equations where $n$ has compact spatial support. The total mass-energy, $M$, and angular momentum, $J$, of the star are taken to be the ADM mass and angular momentum, whereas the total particle number, $N$, and entropy, $S$, of the star are defined by
\be
N=-\int_\Sigma n u^a \nu_a,
\label{N}
\ee
\be
S=-\int_\Sigma sn u^a \nu_a,
\label{S}
\ee
where $\Sigma$ is any Cauchy surface, $\nu^a$ is the unit future-directed normal to $\Sigma$, and the volume element on $\Sigma$ induced by the space-time metric is understood. Note that the integrals in \eqref{N} and \eqref{S} are independent of the choice of Cauchy surface $\Sigma$ by virtue of eqs. \eqref{eq:conservationofn} and \eqref{eq:conservationofs}, i.e., $N$ and $S$ are conserved under dynamical evolution. Note also that the presence of gravitational radiation can contribute to $M$ and $J$ but is assumed not to contribute to $S$, even though in a complete physical theory, gravitational radiation would be expected to contribute to the total entropy.

A fluid star is said to be in {\em dynamic equilibrium} if it is a stationary, axisymmetric solution to the Einstein-fluid equations and, in addition, the velocity flow is circular in the sense that $u^a$ takes the form
\be \label{cirflow}
u^a = (t^a + \Omega \varphi^a)/|v|
\ee
for some function $\Omega$, called the angular velocity, where $t^a$ and $\varphi^a$ are, respectively, the timelike and axial Killing fields, and 
\be
|v|^2 \equiv - (t^a + \Omega \varphi^a)(t_a + \Omega \varphi_a) \, .
\ee
An important consequence of the circular flow assumption is the existence of an additional $t-\varphi$ reflection symmetry when Einstein's equation is satisfied \cite{Carter}.
The requirements of axisymmetry and circular flow have been imposed for technical convenience, but it is not expected that these additional conditions exclude any cases of interest.

A fluid star is said to be in {\em thermodynamic equilibrium} if it is in dynamic equilibrium and, in addition, the total entropy, $S$, of the star \eqref{S} is an extremum at fixed total mass, $M$, angular momentum, $J$, and particle number, $N$, i.e., if $\delta S = 0$ for all first order perturbations of the star that satisfy the linearized constraint equations and for which $\delta M = \delta J = \delta N = 0$. Thus, if a star is in dynamic equilibrium but is not in thermodynamic equilibrium, then its entropy can be increased to first order without changing $M$, $J$, and $N$. As seen above, the perfect fluid equations of motion do not allow changes of $S$ to occur dynamically, so such states of higher entropy are not dynamically accessible for a perfect fluid. However, the perfect fluid equations are expected to be only an idealized description of physically realistic systems. When deviations from perfect fluid behavior are taken into account---i.e., ``dissipative processes," such as heat conduction, viscosity, and diffusion---one would expect that the dynamical evolution would be restricted only by the fundamental conservation laws of $M$, $J$, and $N$. Since $S$ is supposed to measure the ``number of micro-states" associated with the macroscopic fluid description, if a star is thermodynamically unstable, one would expect evolution to higher values of  $S$ to occur, although perhaps on a much longer timescale than typical dynamical timescales. Stars that are in dynamic equilibrium but are not in thermodynamic equilibrium can increase their entropy to first order at fixed $M$, $J$, and $N$, so they would be expected evolve away from their initial perfect fluid equilibrium state---towards a state of higher entropy---when such dissipative processes are taken into account. We will prove in section \ref{thermo} below that the necessary and sufficient condition for a star in dynamic equilibrium to be in thermodynamic equilibrium is that its angular velocity, $\Omega$, redshifted temperature, $\widetilde{T} \equiv T |v|$, and redshifted chemical potential, $\widetilde{\mu} \equiv \mu |v|$, be constant throughout the star. We will also show that perturbations of a star in thermodynamic equilibrium satisfy the first law of thermodynamics in the form
\be
\delta M = \widetilde T \delta S + \widetilde\mu \delta N + \Omega \delta J \, .
\label{1stlaw}
\ee

A fluid star that is in dynamic equilibrium is said to be {\em linearly dynamically stable} if any initially smooth, asymptotically flat solution to the linearized Einstein-fluid equations remains bounded (in some suitable gauge) for all time; otherwise the star is said to be linearly dynamically unstable. A much simpler condition to analyze is ``mode stability": A fluid star that is in dynamic equilibrium is said to be {\em mode stable} if there do not exist any smooth, asymptotically flat, non-pure-gauge linearized solutions with time dependence of the form $\exp(\alpha t)$ with $\rm{Re} (\alpha) > 0$. Obviously, linear dynamic stability implies mode stability. One way of proving mode stability is to find a positive-definite conserved norm on perturbations, since this precludes solutions with exponential growth. The existence of a positive-definite conserved norm on perturbations (together with the similar norms on time derivatives of the perturbations) may also enable one to prove linear dynamic stability---as has been done for the case of perturbations of the Schwarzschild metric \cite{KayWald,dafrod}---but considerable further analysis beyond what is needed to prove mode stability would be required to establish dynamic stability. In particular, from the existence of a conserved, positive-definite norm depending on first derivatives of the perturbation, it may not be straightforward to rule out the existence of perturbations that grow linearly in time. In this paper, we shall obtain a criterion for stability in terms of the existence of a conserved positive-definite norm. The satisfaction of our criterion implies mode stability, but we shall not attempt to show that it implies linear dynamic stability.

A fluid star that is in thermodynamic equilibrium is said to be {\em linearly thermodynamically stable} if $S$ is a strict local maximum at fixed $M$, $J$, and $N$, i.e., if we have $\delta^2 S < 0$ for all first and second order variations that satisfy the constraint equations and keep $M$, $J$, and $N$ fixed to both first and second order. Note that it makes sense to inquire about thermodynamic stability only for stars in thermal equilibrium, since otherwise $S$ can be increased to first order at fixed $M$, $J$, and $N$. Now, it follows directly from the first law of thermodynamics \eqref{1stlaw} that the quantity
\be
{\mathcal E}' \equiv \delta^2 M  - \widetilde T \delta^2 S - \widetilde\mu \delta^2 N - \Omega \delta^2 J 
\label{thermostab}
\ee
is independent of the choice of second order perturbation, and thus is a bilinear quantity in the first order perturbation. For variations for which $\delta^2 M = \delta^2 J = \delta^2 N = 0$, positive-definiteness of ${\mathcal E}'$ is obviously equivalent to negative-definiteness of $\delta^2 S$, since we assume $T>0$ (see \eqref{Tpc}). However, since ${\mathcal E}'$ does not depend upon the choice of second order perturbation, and since second order perturbations can be chosen so as to give $\delta^2 M$, $\delta^2 J$ and $\delta^2 N$ any values one wishes, it follows that thermodynamic stability is equivalent to positivity of ${\mathcal E}'$ for all perturbations for which $\delta M = \delta J = \delta N = 0$ (and hence, by the first law, $\delta S = 0$), but with no restrictions placed on the second order perturbation. 

If a star is thermodynamically stable, then for initial conditions sufficiently close to that of the star---and with the same $M$, $J$, and $N$ as the star---dissipative processes should increase $S$ and thereby drive the state of the system back towards that of the star. Conversely, if one can make $\mathcal E'$ negative for a perturbation with $\delta M = \delta J = \delta N = 0$, then dissipative processes should drive the star away from its thermal equilibrium state. Thus, our notion of thermodynamic stability obtained by examining properties of the entropy functional $S$ on perfect fluid states should be equivalent to notions of thermodynamic stability obtained by examining the dynamics of dissipative fluids, provided that all forms of dissipation are included. Our approach via the entropy functional is much cleaner, simpler, and more general than a direct analysis of a particular form of dissipative equations---and it avoids all of the difficulties associated with obtaining mathematically consistent equations for relativistic dissipative fluids---but it has the disadvantage that one cannot estimate growth timescales, i.e., although one can argue that the star will ``eventually'' evolve to a state of higher entropy, one cannot estimate how long it will take the star to evolve to this state, since this will depend upon the detailed nature of the dissipation.

It should be noted that linear thermodynamic stability implies linear dynamic stability in the sense of mode stability for perturbations with $\delta M = \delta J = \delta N = 0$, since ${\mathcal E}'$ provides a conserved (for perfect fluids) norm on such perturbations. However---just as there is no reason why dynamic equilibrium need imply thermodynamic equilibrium---there is no reason why linear dynamic stability need imply linear thermodynamic stability; the star could be stable with respect to perfect fluid dynamics and yet have nearby states of larger $S$ at fixed $M$, $J$, and $N$. As explained in the previous paragraph, a physically realistic star that is linearly thermodynamically unstable would be expected to have instabilities that are driven by dissipative processes. If the matter that composes the star is well approximated as being a perfect fluid, then the growth timescale of these instabilities would be expected to be much longer than typical dynamical timescales, and there is no reason why the growth need be exponential (as opposed, e.g., to linear) in time. A thermodynamic instability that is not associated with a dynamic instability corresponds to the notion of a {\em secular instability}, defined as an instability that only appears in the presence of a ``dissipative force'' and has a growth rate proportional to the strength of the dissipative force \cite{Hunter, Binney} although this terminology is not always used consistently\footnote{\label{stabterm}\footnotesize In particular, from our perspective, ``gravitational radiation reaction'' is not a ``dissipative force,'' as it is an integral aspect of the overall conservative dynamics of a fluid-gravitational system within the context of general relativity. Thus, we would characterize the CFS instability \cite{Friedman} of rotating fluid stars in general relativity as a dynamic instability, even though the timescale for this instability may be much longer than typical dynamical timescales. Nevertheless, Friedman and Stergioulas \cite{FriedmanBook} apply the term ``secular instability'' to this instability when its timescale is sufficiently long, although they state on p.~251 that this characterization cannot be made precise and ``becomes increasingly blurred as the radiation reaction timescale approaches the period of oscillation of a mode.'' In a footnote to that statement, they propose making this distinction in terms of whether the unstable mode is time symmetric. \medskip}.

It is worth noting that for an infinite, uniform, non-gravitating fluid system, linear thermodynamic stability will hold if and only if entropy always decreases to second order by any first order exchange of energy and number of particles between two fixed volumes of the fluid. This is equivalent to the condition that the matrix of second derivatives of $sn$ with respect to the variables $\rho$ and $n$ be negative definite. Equivalently, linear thermodynamic stability will hold if and only if the matrix of second derivatives of $\rho$ with respect to the variables $sn$ and $n$ is positive definite. This, in turn, is equivalent to the positivity of the determinant and trace of this matrix. This yields the relations
\be\label{localthermostab}
\ba
0 &<  \frac{\partial^2 \rho}{\partial n^2}\frac{\partial^2 \rho}{\partial s^2}-\left(\frac{\partial^2 \rho}{\partial n \partial s} -T \right)^2 
\\
0 &< n^2 \frac{\partial^2 \rho}{\partial n^2}+(1+s^2)\frac{\partial^2 \rho}{\partial s^2}-2sn\left(\frac{\partial^2 \rho}{\partial n \partial s} -T \right),
\ea
\ee
which further restrict the functional form of $\rho$ beyond the relations \eqref{Tpc}. These additional necessary and sufficient conditions for the linear thermodynamic stability of a uniform, non-gravitating fluid are also necessary and sufficient conditions for thermodynamic stability of a star with respect to short length scale perturbations\footnote{\footnotesize Lindblom and Hiscock  \cite{LindblomHiscock} gave only a condition that corresponds to the first of our conditions.\medskip}, since a sufficiently small region of the star can be treated as a homogeneous system, so if relations \eqref{localthermostab} fail, we can increase entropy by energy and particle exchanges within this region, without affecting the global structure of the star. 
However, satisfaction of \eqref{localthermostab} is not sufficient to ensure the linear thermodynamic stability of a star, since thermodynamic instabilities of a global nature may occur.

The contents of the remainder of this paper are as follows. In section \ref{diffco}, the general Lagrangian framework of diffeomorphism covariant theories is reviewed, including the definition of the symplectic form and the derivation of a key identity. In section \ref{thermo}, following Iyer \cite{Iyer}, we will apply these results with respect to only the gravitational part of the Lagrangian to derive a first law of thermodynamics relation between the variation of the ADM mass, $M$, and integrals over the fluid involving the the perturbations to the particle density, $\delta n$, entropy per particle, $\delta s$, and angular momentum per particle, $\delta j$. This formula will then be used to prove that the necessary and sufficient condition for a star in dynamic equilibrium to be in thermodynamic equilibrium is that its angular velocity, redshifted temperature, and redshifted chemical potential be constant throughout the star. In that case, the first law of thermodynamics takes the form \eqref{1stlaw}. The results of section \ref{thermo} do not require a Lagrangian framework for the description of the perfect fluid.

In section \ref{Lagrangianform}, we introduce a Lagrangian framework for perfect fluids, which is of the type used by Friedman \cite{Friedman}. To formulate this precisely, we introduce in subsection \ref{LagrangianformA} a $4$-dimensional ``fiducial manifold'' $M'$ that is diffeomorphic to the spacetime manifold $M$. On $M'$ are defined a fixed closed $3$-form $\bN'$ and scalar function $s'$, satisfying $d \bN' = 0$ and $d(s' \bN') = 0$. The dynamical variables of the theory then consist of a spacetime metric $g_{ab}$ on $M$ and a diffeomorphism $\chi: M' \to M$. The entropy per particle, $s$, on spacetime is taken to be the pushforward of $s'$ to $M$ under $\chi$, whereas $n$ and the fluid $4$-velocity $u^a$ on spacetime are constructed from the spacetime metric and the pushforward, $\bN$, of $\bN'$ to $M$. Linearized perturbations are then described by a perturbed spacetime metric $\delta g_{ab}$ and a vector field $\xi^a$, known as the {\em Lagrangian displacement vector field}, representing the infinitesimal change in $\chi$. An arbitrary Einstein-fluid solution can be represented in this framework by suitably choosing $s'$ and $\bN'$. However, once $s'$ and $\bN'$ are chosen, they are required to remain fixed, so consideration is thereby restricted in the Lagrangian framework only to perturbations that correspond to ``moving fluid elements around,'' i.e., the only kinematically allowed changes in entropy per particle and particle current $3$-form are of the form $-\lie_\xi s$ and $-\lie_\xi \bN$, respectively. The symplectic current of the Einstein-fluid system also is obtained in subsection \ref{LagrangianformA}. This symplectic current is then used in subsection \ref{LagrangianformB} to construct the phase space. It is shown that specification of a point in phase space is equivalent to the specification on a Cauchy surface $\Sigma$ of a spatial metric $h_{ij}$, the usual gravitational canonical momentum $\pi^{ij}$, a diffeomorphism from the manifold of ``fiducial flowlines'' (the manifold of integral curves of the vector field on $M'$ that annihilates $\bN'$ through contraction) to $\Sigma$, and the spatial components of the fluid $4$-velocity, $u^i$. A Hilbert space structure, $\mK$, on perturbations is also defined in subsection \ref{LagrangianformB}. The phase space includes a subspace of ``trivial displacements,'' which produce no changes in the physical state of the fluid\footnote{\footnotesize The symplectic form is degenerate on trivial displacements of the form $\xi^a = f u^a$ for any function $f$, so the trivial displacements of this form are automatically eliminated from the phase space. However, there remain additional trivial displacements that are not degeneracies of the symplectic form.\medskip}. The properties of these trivial displacements are analyzed in subsection \ref{LagrangianformC}.

The canonical energy, $\mathcal E$, is defined in section \ref{canonicalenergy}. The canonical energy is conserved in the sense of taking the same value on all asymptotically flat Cauchy surfaces.  However, if the perturbation asymptotically approaches a stationary perturbation at late retarded times, then $\mathcal E$ has a positive net flux at null infinity \cite{Friedman,HollandsWald}.  In that case, $\mathcal E$ can only decrease if evaluated on a hypersurface terminating at null infinity at late retarded times. In order to be able to use the positivity of $\mathcal E$ as a criterion for dynamic stability, it is essential that $\mathcal E$ be degenerate precisely on the perturbations that are physically stationary. Namely, if $\mathcal E$ were degenerate on perturbations that are not physically stationary, its non-negativity cannot guarantee stability, since the perturbations on which it is degenerate could grow exponentially with time. Conversely, if $\mathcal E$ were non-degenerate on physically stationary perturbations, then its failure to be positive cannot guarantee instability, since it could take negative values on physically stationary perturbations, which are manifestly stable. We show that in order to make $\mathcal E$ degenerate precisely on the physically stationary perturbations, it is necessary to restrict the action of $\mE$ to the subspace $\mV$ of perturbations that are symplectically orthogonal to all of the trivial displacements. Perturbations in $\mV$ must satisfy $\Delta \alpha = 0$, where $\Delta \alpha$ denotes the Lagrangian change in circulation, and---depending on the properties of the background star---some additional restrictions may also hold on perturbations in $\mV$. If $\mE \geq 0$ on $\mV$, then mode stability holds for perturbations in $\mV$, whereas if $\mE < 0$ for some perturbation in $\mV$, then this perturbation cannot asymptotically approach a stationary final state, indicating dynamic instability. 

As shown by Friedman \cite{Friedman}, if $\Omega \not\equiv 0$, one can find non-axisymmetric perturbations in $\mV$ for which $\mathcal E < 0$, thus showing that all rotating stars are dynamically unstable (the CFS instability). As also shown by Friedman, in the non-axisymmetric case, the restriction to perturbations in $\mV$ is not a physical restriction---at least for certain background solutions---in that it can be imposed on a general perturbation by addition of a perturbation arising from a trivial displacement. In other words, in the non-axisymmetric case, restriction to $\mV$ corresponds to merely a ``gauge choice'' on the choice of Lagrangian displacement rather than a physical restriction on the perturbation. However, for axisymmetric perturbations, restriction to $\mV$ is a physical restriction on perturbations. Thus, in the axisymmetric case, positivity of $\mE$ on $\mV$ directly shows mode stability only on this restricted class of perturbations. However, we prove in section \ref{canonicalenergy} that, given an arbitrary perturbation---possibly not even describable in the Lagrangian framework---its second time derivative yields a perturbation in $\mV$. Consequently, in the axisymmetric case, mode stability for perturbations in $\mV$ implies mode stability for all perturbations, and positivity of $\mE$ on $\mV$ is necessary and sufficient for dynamic stability with respect to {\em all} axisymmetric perturbations.

Finally, thermodynamic stability of stars in thermodynamic equilibrium is analyzed in section VI. It is shown that a necessary condition for linear thermodynamic stability  is positivity of the canonical energy $\mathcal E_r$---defined with respect to the Killing field $v^a$ to which the fluid $4$-velocity is proportional---on the space of  Lagrangian perturbations having $\delta J = 0$, where $J$ is the total (ADM) angular momentum. It is easily seen that if $\Omega \neq 0$, non-axisymmetric perturbations can be found for which $\mathcal E_r < 0$, so all rotating stars are thermodynamically unstable, although not necessarily on a physically interesting timescale (see footnote \ref{lh}). For axisymmetric perturbations, we have $\mathcal E_r = \mathcal E$, so this necessary condition for thermodynamic stability in the axisymmetric case is positivity of $\mE$ (on all perturbations with $\delta J = 0$, not merely the ones in $\mV$).

Our notation and conventions will generally follow \cite{Waldbook}. Latin indices from the early part of the alphabet ($a,b,c,\dots$) will denote abstract spacetime indices, whereas Latin indices from the mid-alphabet ($i,j,k,\dots$) will denote abstract spatial indices associated with a Cauchy surface. 
Bold typeface will indicate that differential form indices on spacetime have been omitted, e.g., $\bN$ denotes the tensor field $N_{abc}=N_{[abc]}$.

\section{Lagrangian Framework for Diffeomorphism Covariant Theories} \label{diffco}

In this section, we will review the basic constructions of the Lagrangian framework for diffeomorphism covariant theories. We will apply these results to the vacuum Einstein-Hilbert Lagrangian in section \ref{thermo} (even though we will derive results there that are valid when perfect fluid matter is present) and we will apply these constructions to the Einstein-fluid Lagrangian in section \ref{Lagrangianform}. We refer the reader to \cite{LeeWald,IyerWald,IyerWald2,SeifertWald} for a much more complete account.

We assume that we are given a diffeomorphism covariant Lagrangian $4$-form $L_{abcd}$ that is locally and covariantly constructed out of the spacetime metric, $g_{ab}$, and other tensor fields ${\chi^{a_1 \dots a_k}}_{b_1 \dots b_l}$. We collectively denote the dynamical fields as $\phi = (g_{ab}, {\chi^{a_1 \dots a_k}}_{b_1 \dots b_l})$ and suppress all spacetime indices on the dynamical fields for the remainder of this section.

Variation of the Lagrangian yields
\begin{equation}
  \label{eq:L}
  \delta \bL = \bE \cdot \delta \phi + d \btheta(\phi; \delta \phi) \, ,
\end{equation}
where $\bE=0$ are the field equations\footnote{\footnotesize Besides the 4-form indices, the additional indices of $\bE$ that are dual to 
$\phi$ have been suppressed; the ``$\cdot$'' in \eqref{eq:L} denotes that those additional indices are fully contracted into all tensor indices of $\delta\phi$. For example, in vacuum general relativity we have $\bE \cdot \delta \phi=-\frac{1}{16 \pi} G^{ab} \delta g_{ab} \bepsilon$.\medskip}, and the {\em symplectic potential $3$-form} $\btheta$ corresponds to the boundary term that would arise if the variation were performed under an integral sign. The quantity $\delta \phi$ may be formally viewed as a vector---which we denote as $\delta \phi^A$---in the tangent space at $\phi$ of the (infinite dimensional) space of field configurations $\mathcal F$. We will not attempt to define a manifold structure on $\mathcal F$, so our tensor operations on $\mathcal F$ below should be viewed as merely ``formal''; they will be used only to motivate various definitions.

At each $\phi \in {\mathcal F}$ , we obtain a linear map from vectors, $\delta \phi^A$, into numbers by integration of the $3$-form $\btheta$ over a Cauchy surface $\Sigma$. We can interpret this linear map as defining a $1$-form field $\Theta_A$ on $\mathcal F$
\be
\Theta_A \delta \phi^A = \int_\Sigma \btheta(\phi; \delta \phi) \, .
\ee
The symplectic form is defined by
\be \label{WfromTheta}
W_{AB} = ({\mathcal D}  \Theta)_{AB} \, ,
\ee
where $\mathcal D$ denotes the ``$d$" (exterior derivative) operation on forms on $\mathcal F$. In order to evaluate $W_{AB} (\delta_1 \phi)^A (\delta_2 \phi)^B$ at a point $\phi \in {\mathcal F}$, we obviously need only specify the vectors $\delta_1 \phi^A$ and $\delta_2 \phi^B$ at $\phi$. However, it is useful to evaluate $W_{AB} (\delta_1 \phi)^A (\delta_2 \phi)^B$ via the formula\footnote{\footnotesize The formulas  \eqref{WfromTheta} and \eqref{sympform} can formally be seen to be equivalent by introducing an arbitrary derivative operator, $\nabla_A$, on $\mathcal F$ and expanding both formulas in terms of $\nabla_A$ using the usual Lie derivative and exterior derivative expressions.\medskip}
\be
W_{AB} (\delta_1 \phi)^A (\delta_2 \phi)^B = {\mathcal L}_{\delta_1 \phi} (\Theta_B \delta_2 \phi^B) - {\mathcal L}_{\delta_2 \phi} (\Theta_A \delta_1 \phi^A) - \Theta_A ({\mathcal L}_{\delta_1 \phi} \delta_2 \phi)^A \, ,
\label{sympform}
\ee
where $\mathcal L$ denotes the Lie derivative on $\mathcal F$. Use of this formula requires us to extend the definitions of $\delta_1 \phi^A$ and $\delta_2 \phi^B$ off of $\phi$, which can be done in an arbitrary manner. For most purposes, it is convenient to choose the extended vector fields $\delta_1 \phi^A$ and $\delta_2 \phi^B$ so as to commute, which can be done by choosing a 2-parameter family of field configurations $\phi(\lambda_1, \lambda_2)$ and choosing $\delta_1 \phi = \partial \phi/\partial \lambda_1$ and $\delta_2 \phi = \partial \phi/\partial \lambda_2$. In that case, the last term in \eqref{sympform} vanishes. However, in the case where the field variations $\delta_1 \phi^A$ and/or $\delta_2 \phi^B$ arise from the action of an infinitesimal diffeomorphism on $\phi$, it is convenient to choose the extension of $\delta_1 \phi^A$ and/or $\delta_2 \phi^B$ to be the field variation produced by the action of the same infinitesimal diffeomorphism, in which case $\delta_1 \phi^A$ and $\delta_2 \phi^B$ need not commute, and the last term in \eqref{sympform} must be kept. 

Equation \eqref{sympform} corresponds to the formula
\be \label{sympform2}
W_{AB} (\delta_1 \phi)^A (\delta_2 \phi)^B = \int_\Sigma \bomega(\phi; \delta_1 \phi, \delta_2 \phi) 
\ee
where the {\em symplectic current $3$-form} $\bomega$ on spacetime is defined by
\be \label{sympcurrent}
\bomega(\phi; \delta_1 \phi, \delta_2 \phi) = \delta_1 \btheta (\phi; \delta_2 \phi) - \delta_2 \btheta(\phi; \delta_1 \phi)
- \btheta(\phi; \delta_1 \delta_2 \phi - \delta_2 \delta_1 \phi)
\ee
and $\delta_1$ and $\delta_2$ denote the variation of quantities induced by the field variations $\delta_1 \phi$ and $\delta_2 \phi$ respectively. It follows immediately from \eqref{eq:L} that
\be
d \bomega(\phi; \delta_1 \phi, \delta_2 \phi) = \delta_2 \bE \cdot \delta_1 \phi -  \delta_1 \bE \cdot \delta_2 \phi \, ,
\label{sympcons}
\ee
so $\bomega$ is closed whenever $\delta_1 \phi$ and $\delta_2 \phi$ satisfy the linearized equations of motion $\delta_1 \bE = \delta_2 \bE = 0$. Consequently, if the linearized equations of motion hold, then $W_{AB} (\delta_1 \phi)^A (\delta_2 \phi)^B$ is conserved in the sense that it takes the same value if the integral defining this quantity is performed over the surface $\Sigma'$ rather than $\Sigma$, where $\Sigma'$ and $\Sigma$ bound a compact region. For asymptotically flat space times, $W_{AB} (\delta_1 \phi)^A (\delta_2 \phi)^B$ takes the same value on any two asymptotically flat Cauchy surfaces $\Sigma$ and $\Sigma'$ provided that $\delta_1 \phi$ and $\delta_2 \phi$ satisfy the linearized equations of motion and have suitable fall-off at infinity.

For a diffeomorphism covariant Lagrangian, the Noether current $3$-form, $\boldsymbol{\mathcal J}_X$, on spacetime associated with an arbitrary vector field $X^a$ is defined by
\be
\boldsymbol{\mathcal{J}}_X = \btheta(\phi; \pounds_X \phi) - i_X \bL,
\ee
where $\pounds$ denotes the spacetime Lie derivative and $i_X$ denotes contraction of $X^a$ into the first index of a differential form. A simple calculation \cite{IyerWald} shows that
the first variation of $\boldsymbol{\mathcal{J}}_X$ (with $X^a$ fixed, i.e., unvaried) satisfies
\be
\delta \boldsymbol{\mathcal{J}}_X = - i_X \left[\bE(\phi) \cdot \delta \phi\right] + \bomega\left(\phi;\delta \phi, \pounds_X  \phi\right) + d (i_X \btheta) \, ,
\label{dJ}
\ee
where, in this formula, it has not been assumed that $\phi$ satisfies
the field equations nor that $\delta \phi$ satisfies the linearized
field equations. Furthermore, it can be shown \cite{IyerWald2} that
$\boldsymbol{\mathcal{J}}_X$ can be written in the form
\be
\boldsymbol{\mathcal{J}}_X = \bC_X + d \bQ_X \, ,
\label{noecon}
\ee
where $\bQ_X$
is the Noether charge and $\bC_X\equiv \bC_a X^a$ with $\bC_a = 0$ being the constraint equations of the theory \cite{SeifertWald}.
Combining eqs.~(\ref{dJ}) and~(\ref{noecon}), we obtain the fundamental identity
\be
\bomega(\phi;\delta \phi, \pounds_X \phi) = i_X (\bE(\phi) \cdot \delta \phi) + \delta \bC_X(\phi)
+ d \left[ \delta \bQ_X(\phi) - i_X \btheta(\phi;\delta \phi) \right] \, .
\label{fundid}
\ee
It should be emphasized that eq.~(\ref{fundid}) holds for arbitrary $X^a$, $\phi$, and $\delta \phi$. In particular, $\phi$ need not
satisfy the equations of motion and $\delta \phi$ need not satisfy the linearized equations of motion.

One immediate consequence of \eqref{fundid} is the gauge invariance of the symplectic form. If $\phi$ satisfies the equations of motion, $\bE(\phi) = 0$, if $\delta \phi$ satisfies the linearized constraints, $\delta \bC_a = 0$, and if $X^a$ is of compact support (or vanishes sufficiently rapidly at infinity and/or any boundaries), integration of \eqref{fundid} over a Cauchy surface $\Sigma$ yields
\be
W_{AB} (\delta \phi)^A (\pounds_X \phi)^B  = 0 \, .
\label{giW}
\ee
Consequently, the value of $W_{AB} (\delta_1 \phi)^A (\delta_2 \phi)^B$ is unchanged if either $\delta_1 \phi$ or $\delta_2 \phi$ is altered by a gauge transformation $\delta \phi \to \delta \phi + \lie_X \phi$ with $X^a$ of compact support.

Another very important application of the identity \eqref{fundid} concerns the case where $X^a$ approaches a nontrivial asymptotic symmetry rather than being of compact support, in which case we can derive a formula for the
Hamiltonian, $H_X$, conjugate to the notion of ``time translations'' defined by $X^a$, and, thereby, a
definition of ADM-type conserved quantities. Consider asymptotically flat spacetimes with one asymptotically flat ``end". Integrating \eqref{fundid} over a Cauchy surface $\Sigma$, we obtain\footnote{\footnotesize Additional ``boundary terms'' would appear in  \eqref{ham0} if $\Sigma$ terminated at a bifurcate Killing horizon or if there were additional asymptotically flat ends.\medskip}
\be
W_{AB} (\delta \phi)^A (\pounds_X \phi)^B = \int_\Sigma \left[i_X (\bE(\phi) \cdot \delta \phi) + \delta \bC_X(\phi) \right] +
\int_{S_\infty} \left[ \delta \bQ_X(\phi) - i_X \btheta(\phi;\delta \phi) \right] \, ,
\label{ham0}
\ee
where the second integral is taken over a $2$-sphere $S$ that limits to infinity. Suppose that, in this limit, we have
\be
\lim_{S \to S_\infty} \int_S i_X \btheta(\phi;\delta \phi) = \lim_{S \to S_\infty} \delta \int_S i_X \boldsymbol B(\phi)
\label{B}
\ee
for some (``non-covariant'') $3$-form $\boldsymbol B$ constructed from $\phi$ and the background asymptotic structure near infinity. Then, if $\phi$ satisfies the
equations of motion, $\bE(\phi) = 0$---but $\delta \phi$ is {\em not} required to satisfy the linearized equations of motion---we have
\be
W_{AB} (\delta \phi)^A (\pounds_X \phi)^B = \delta H_X
\label{ham1}
\ee
where
\be
H_X \equiv \int_\Sigma \bC_X + \int_{S_\infty} [\bQ_X - i_X \boldsymbol B] \, .
\label{ham2}
\ee

Writing $\delta H_X = (\delta \phi)^A {\mathcal D}_A (H_X)$, we may rewrite \eqref{ham1} as
\be
W_{AB} (\pounds_X \phi)^B = {\mathcal D}_A (H_X) \, .
\label{ham3}
\ee
We now pass from field configuration space, $\mathcal F$, to phase space, $\mathcal P$, by factoring by the degeneracy orbits of $W_{AB}$, as described in \cite{LeeWald}. On $\mathcal P$, $W_{AB}$ is well defined and, by construction, is nondegenerate. Let $W^{AB}$ denote the inverse of $W_{AB}$, so that $W^{AB} W_{BC} = {\delta^A}_C$ where ${\delta^A}_C$ denotes the identity map on $\mathcal P$. Then, we have
\be
(\pounds_X \phi)^A = W^{AB} {\mathcal D}_B (H_X) \, ,
\label{ham4}
\ee
which is the usual form of Hamilton's equations of motion on a symplectic manifold. Thus, if both the asymptotic conditions on $\phi$ and the asymptotic behavior of $X^a$ are such that a $3$-form $\boldsymbol B$ satisfying \eqref{B} exists, \eqref{ham2} yields a Hamiltonian conjugate to the notion of ``time translations'' defined by $X^a$. Note that when evaluated on solutions, $\bC_X = 0$, so $H_X$ is purely a ``surface term"
\be
\left.H_X\right|_{\bE=0} = \int_{S_\infty} [\bQ_X - i_X \boldsymbol B] \, .
\label{ham5}
\ee

In the case where $X^a$ is asymptotic to a time translation at infinity, a $\boldsymbol B$ satisfying \eqref{B} can be found \cite{IyerWald}, and
\eqref{ham5} with $X^a = t^a$ defines the ADM mass
\be
M = \int_{S_\infty} [\bQ_t - i_t \boldsymbol B] \, .
\ee
In the case where $X^a$ is asymptotic to a rotation tangent to $\Sigma$ at infinity and $S$ is chosen so that $X^a$ is tangent to $S$, the pullback of $i_X \btheta$ to $S$ vanishes, and \eqref{ham5} with $X^a = \varphi^a$ and $\boldsymbol B = 0$ defines minus\footnote{\footnotesize The map
$X \rightarrow H_X$ is a linear functional and thus is a ``covector''. To get the energy-momentum vector, we have to ``raise the index'' with the Minkowskian metric. This accounts for the relative minus sign between $M$ and  $J$  in their definitions in terms of $H_X$. \medskip} the ADM angular momentum
\be \label{JADM}
J = -  \int_{S_\infty} \bQ_\varphi \, .
\ee

Finally, let us return to \eqref{ham0} in the case where $\phi$ has a time translation symmetry, i.e., $\pounds_t \phi = 0$ for a vector field $t^a$ that approaches a time translation at infinity. We further assume that the equations of motion, $\bE(\phi) = 0$, hold in a neighborhood of infinity, but we do not assume that they hold in the interior of the spacetime. We similarly assume that $\delta \phi$ satisfies the linearized constraints near infinity, but do not assume that these hold in the interior of the spacetime, nor do we make any symmetry assumptions on $\delta \phi$. Then the left side of \eqref{ham0} vanishes for $X^a = t^a$, and the surface integral on the right side simply yields $\delta M$. Thus we obtain
\be
\delta M = - \int_\Sigma \left[i_t (\bE(\phi) \cdot \delta \phi) + \delta \bC_t(\phi) \right]  \, .
\label{ham6}
\ee
As we shall see in the next section, this formula yields the first law of thermodynamics for fluid stars.

\section{The First Law of Thermodynamics and Thermodynamic Equilibrium} \label{thermo}

Following Iyer \cite{Iyer}, we now apply the results of the previous section to the {\em vacuum} Einstein-Hilbert Lagrangian
\be
L_{abcd}=\frac{1}{16\pi}R\epsilon_{abcd} \, ,
\ee
to obtain the first law of thermodynamics for perfect fluid stars\footnote{\footnotesize Bardeen, Carter, and Hawking \cite{BCH} derived the first law only for the case of stationary and axisymmetric perturbations. The first law for general perturbations is implicit in Corollary 6.6 of Schutz and Sorkin \cite{SchutzSorkin}. Here we follow Iyer's \cite{Iyer} derivation.\medskip}.  For this Lagrangian, the equations of motion tensor field (see \eqref{eq:L}) is
\be
{E^{ab}}_{cdef} = -\frac{1}{16 \pi} G^{ab} \epsilon_{cdef}\,,
\label{Eeom}
\ee
the constraint 3-form $\bC_X$ is
\be
(C_X)_{abc} = \frac{1}{8 \pi} X^d {G_d}^e \epsilon_{eabc} \, ,
\ee
and the Noether charge $2$-form is
\be
(Q_X)_{ab} = - \frac{1}{16 \pi} \nabla_c X_d {\epsilon^{cd}}_{ab} \, . 
\label{ncgr}
\ee
Let $g_{ab}$ be any asymptotically flat, stationary (i.e., $\pounds_t g_{ab} = 0$ for some $t^a$ that is timelike near infinity) metric that is a vacuum solution of Einstein's equation near infinity. Let $\delta g_{ab}$ be any asymptotically flat perturbation of $g_{ab}$ that satisfies the linearized vacuum Einstein equation near infinity. Then
eq.~\eqref{ham6} yields
\be
\delta M = \frac{1}{8 \pi} \int_\Sigma t^a\left[\frac{1}{2} G^{bc} \delta g_{bc}  \epsilon_{adef} - 
 \delta \left({G_a}^b \epsilon_{bdef}\right)\right] \, .
\ee

We now further assume that, in addition to being stationary and asymptotically flat, $g_{ab}$ is axisymmetric (i.e., $\pounds_\varphi g_{ab} = 0$ for some space like $\varphi^a$ with closed orbits) and satisfies
\be
G_{ab} = 8 \pi T_{ab}
\ee
for some $T_{ab}$ of the perfect fluid form \eqref{perflu} having compact spatial support. We further assume that the $4$-velocity $u^a$ appearing in \eqref{perflu} has circular flow \eqref{cirflow}. In other words, we assume that $g_{ab}$ is the metric of a fluid star in dynamic equilibrium, as defined in the Introduction. In addition, we assume that $\delta g_{ab}$ satisfies the linearized Einstein-fluid equations
\be
\delta G_{ab} = 8 \pi \delta T_{ab}\,,
\ee
where $\delta T_{ab}$ takes the form of a perturbed perfect fluid \eqref{perflu} of compact spatial support. However, we impose no symmetry conditions on $\delta g_{ab}$, nor do we impose any conditions on the perturbed $4$-velocity $\delta u^a$. For convenience, we choose $\Sigma$ to be axisymmetric in the sense that $\varphi^a$ is tangent to $\Sigma$, so that the pullback of $\varphi^a\epsilon_{abcd}$ to $\Sigma$ vanishes. Using the circular flow condition \eqref{cirflow} of the background spacetime to write
\be
t^a = |v|u^a - \Omega \varphi^a \, ,
\ee 
we have\footnote{\footnotesize The vector fields $t^a$ and $\varphi^a$ are fixed (``field independent''), so $\delta t^a = \delta \varphi^a = 0$.\medskip}
\be \label{firstlaw}
\ba
  \delta M &= \int_\Sigma t^a \left[ \frac{1}{2}T^{bc}\delta g_{bc}\epsilon_{adef} - \delta\left(T_a^{\phantom{a}b}\epsilon_{bdef}\right) \right] \\
  &= \int_\Sigma \left\{ |v| u^a \left[ \frac{1}{2}T^{bc}\delta g_{bc}\epsilon_{adef} - \delta\left(T_a^{\phantom{a}b}\epsilon_{bdef}\right) \right] +  \Omega   \delta\left(\varphi^a T_a^{\phantom{a}b}\epsilon_{bdef}\right)  \right\}.
\ea
\ee

We define the particle current $3$-form, ${\bN}$, by
\begin{equation} \label{Ndefn}
  N_{def} = nu^a\epsilon_{adef}\,,
\end{equation}
the entropy current $3$-form, ${\bS}$, by
\be
 S_{def} = s N_{def}.
\ee
and the angular momentum current $3$-form, ${\bJ}$, by
\be
  J_{def} = \varphi^a T_a^{\phantom{a}b}\epsilon_{bdef}\,.
\ee
The following lemma allows us to write the first two terms in \eqref{firstlaw} a convenient form:

\begin{lemma}\label{fluididentity}
For any smooth one parameter family of Einstein-perfect fluid field configurations (not necessarily satisfying the field equations), we have
\be \label{fluidid}
u^a(\lambda) \left[ \frac{1}{2}T^{bc}(\lambda)\frac{d g_{bc}}{d\lambda}\epsilon_{adef}(\lambda) - \frac{d\left(T_a^{\phantom{a}b}\epsilon_{bdef}\right)}{d\lambda} \right] = \mu(\lambda) \frac{d N_{def}}{d\lambda} + T(\lambda) \frac{dS_{def}}{d\lambda},
\ee
\end{lemma}

\begin{proof}
Using the perfect fluid form of the stress energy \eqref{perflu}, the relation \eqref{gibbsduhem}, and the local first law of thermodynamics \eqref{loc1st} in the alternative form
\be
dp=n d\mu + nsdT,
\ee
we calculate
\be
\ba
\frac{d}{d\lambda}&\left(u^a  T_a^{\phantom ab} \epsilon_{bdef} \right) = \frac{d}{d\lambda}\left[-(\rho+p) u^a \epsilon_{adef} + p u^a \epsilon_{adef} \right] \\
&= \frac{d}{d\lambda}\left[-(\mu+sT) n u^a \epsilon_{adef} + p u^a \epsilon_{adef} \right] \\
&= -\mu \frac{d}{d\lambda}\left(n u^a \epsilon_{adef}\right)-T\frac{d}{d\lambda}\left(s n u^a \epsilon_{adef}\right) - \left(\frac{d\mu}{d\lambda}+s\frac{dT}{d\lambda}\right)n u^a \epsilon_{adef}+ \frac{d}{d\lambda} \left( p u^a \epsilon_{adef} \right) \\
&= -\mu \frac{dN_{def}}{d\lambda}-T\frac{dS_{def}}{d\lambda} - \frac{dp}{d\lambda} u^a \epsilon_{adef}+ \frac{d}{d\lambda} \left( p u^a \epsilon_{adef} \right) \\
&= -\mu \frac{dN_{def}}{d\lambda}-T\frac{dS_{def}}{d\lambda} +p\frac{du^a}{d\lambda} \epsilon_{adef} +\frac{1}{2}p g^{bc} \frac{d g_{bc}}{d\lambda}  u^a \epsilon_{adef} \, .
\ea
\ee
Furthermore, from 
\be
u_a \frac{du^a}{d\lambda} = -\frac{1}{2}u^b u^c \frac{dg_{bc}}{d\lambda} \, ,
\ee
we have
\be
\ba
\frac{du^a}{d\lambda}T_a^{\phantom ab}\epsilon_{bdef} &= (\rho+p) u^b \epsilon_{bdef} u_a \frac{du^a}{d\lambda} + p \frac{du^a}{d\lambda} \epsilon_{adef}\\
&= -\frac{1}{2}(\rho+p) u^b u^c \frac{d g_{bc}}{d\lambda} u^a \epsilon_{adef} + p \frac{du^a}{d\lambda} \epsilon_{adef} \, .
\ea
\ee
Thus, we obtain
\be
\ba
u^a\frac{d}{d\lambda}&\left(  T_a^{\phantom ab} \epsilon_{bdef} \right) = \frac{d}{d\lambda}\left(u^a  T_a^{\phantom ab} \epsilon_{bdef} \right) - \frac{du^a}{d\lambda}T_a^{\phantom ab}\epsilon_{bdef} \\
&= -\mu \frac{dN_{def}}{d\lambda}-T\frac{dS_{def}}{d\lambda} +\frac{1}{2}p g^{bc} \frac{d g_{bc}}{d\lambda}  u^a \epsilon_{adef}  +\frac{1}{2}(\rho+p) u^b u^c \frac{d g_{bc}}{d\lambda} u^a \epsilon_{adef} \\
&= -\mu \frac{dN_{def}}{d\lambda}-T\frac{dS_{def}}{d\lambda} +\frac{1}{2}T^{bc} \frac{d g_{bc}}{d\lambda}  u^a \epsilon_{adef},
\ea
\ee
which completes the proof.
\end{proof}

Defining the redshifted temperature and chemical potential by
\be
\begin{aligned}
\widetilde T &= T |v|\,,\\
\widetilde\mu &= \mu  |v|\,,
\end{aligned}
\ee
we can now write \eqref{firstlaw} in the form
\be \label{firstlaw2}
\delta M = \int_\Sigma \left( \widetilde\mu \delta \bN + \widetilde T \delta \bS + \Omega\delta \bJ \right) \, .
\ee
Equation \eqref{firstlaw2} is our desired form of the first law of thermodynamics, which 
holds for arbitrary perturbations off of a fluid star in dynamic equilibrium. If a black hole also was present, there would be an additional contribution to \eqref{firstlaw2} from the black hole horizon \cite{Iyer}.

The $3$-forms $\bN$, $\bS$, and $\bJ$ are dual, respectively to the particle number current $n u^a$, the entropy current, $sn u^a$, and the angular momentum current $\varphi^b T_b^{\phantom ba}$. The first two currents are conserved by \eqref{eq:conservationofn} and \eqref{eq:conservationofs}, and, in the axisymmetric case (i.e., if $\varphi^a$ is a Killing field) we also have $\nabla_a (\varphi^b T_b^{\phantom ba}) = 0$. Thus, $\bN$ and $\bS$ are always closed
\be
d \bN = 0 \, , \quad \quad  d \bS = 0,
\ee
and in the axisymmetric case, we also have
\be
d \bJ = 0 \, .
\ee
By \eqref{N} and \eqref{S}, the total number of particles, $N$, and the total entropy, $S$, are given by
\be
N = \int_\Sigma \bN  \, , \quad \quad  S = \int_\Sigma \bS \, .
\ee
In addition, for an axisymmetric metric $g_{ab}$ that is a vacuum solution to Einstein's equation near infinity, the ADM angular momentum \eqref{JADM} is given by
\be
\ba
J &= -  \int_{S_\infty} \bQ_\varphi = - \int_\Sigma d \bQ_\varphi   \\
&= - \frac{1}{8 \pi} \int_\Sigma \nabla_e (\nabla^{[e} \varphi^{d]}) \epsilon_{dabc}  \\
&= \frac{1}{8 \pi} \int_\Sigma \varphi^e {R_e}^d \epsilon_{dabc}  \\
&= \int_\Sigma \varphi^e {T_e}^d \epsilon_{dabc}  \\
&= \int_\Sigma \bJ \, ,
\label{Jeq}
\ea
\ee
where \eqref{ncgr} was used in the second line.
Furthermore, for any first order (possibly non-axisymmetric) perturbation $\delta g_{ab}$ of an axisymmetric metric, we have
\be
\delta J = \int_\Sigma \delta \bJ \, .
\label{Jeq2}
\ee
Namely, we can write a general perturbation $\delta g_{ab}$ as a sum of perturbations satisfying $\pounds_\varphi \delta g_{ab} = i m \delta g_{ab}$. Perturbations with $m = 0$ are axisymmetric and thus satisfy \eqref{Jeq}. Perturbations with $m \neq 0$ satisfy $\delta J = 0$ and $\int_\Sigma  \delta \bJ = 0$. Thus, all first order perturbations satisfy \eqref{Jeq2}. However, at second order, in general, we have
\be
\delta^2 J \neq \int_\Sigma \delta^2 \bJ 
\ee
since gravitational radiation can now contribute to the ADM angular momentum.

As stated in the introduction, a fluid star in dynamic equilibrium is said to be in {\em thermodynamic equilibrium} if and only if $\delta S=0$ with respect to all perturbations that satisfy the linearized Einstein constraint equations and for which $\delta M = \delta N = \delta J = 0$. We conclude this section with the following theorem.

\begin{theorem}
A dynamic equilibrium configuration is in thermodynamic equilibrium if and only if $\widetilde T$, $\widetilde\mu$, and $\Omega$ are uniform throughout the star.
\end{theorem}

\begin{proof} The proof of the ``if'' part is entirely straightforward: If $\widetilde T$, $\widetilde\mu$, and $\Omega$ are constant throughout the star, the first law \eqref{firstlaw2} reduces to
\begin{align} 
\delta M &= \widetilde\mu \int_\Sigma \delta \bN + \widetilde T \int_\Sigma \delta \bS 
+ \Omega \int_\Sigma \delta \bJ  \nonumber  \\
&= \widetilde\mu \delta N + \widetilde T \delta S + \Omega \delta J \, ,
\label{firstlaw3}
\end{align}
from which it is immediately obvious that $\delta S = 0$ for any perturbation with $\delta M = \delta N = \delta J = 0$. 

The proof of the ``only if" part would be similarly straightforward if one could freely choose the quantities 
$\overline {\delta\bN}$, $\overline {\delta\bS}$, and $\overline {\delta\bJ}$, where an overline denotes the pullback to $\Sigma$. If that were the case, it follows immediately from \eqref{firstlaw2} that if at least one of $\widetilde T$, $\widetilde\mu$, or $\Omega$ were nonuniform, one could find a perturbation with $\delta S = \delta N = \delta J = 0$ but with $\delta M \neq 0$. Since $\widetilde T > 0$ (see \eqref{Tpc}), one could then find a second perturbation with 
$\overline {\delta\bN} = \overline {\delta\bJ} = 0$ (and, hence, $\delta N = \delta J = 0$) but $\overline {\delta\bS} \neq 0$ in such a way that $\delta S \neq 0$ and $\delta M \neq 0$. By combining these two perturbations one can then obtain a perturbation with $\delta M = \delta N = \delta J = 0$ but $\delta S \neq 0$, thereby showing that a dynamic equilibrium star with nonuniform $\widetilde T$, $\widetilde\mu$, or $\Omega$ cannot be in thermodynamic equilibrium.

However, the Einstein constraint equations impose nontrivial restrictions on the allowed perturbations, and, {\em a priori}, it is far from obvious that $\overline {\delta\bN}$, $\overline {\delta\bS}$, and $\overline {\delta\bJ}$ can be chosen freely. Nevertheless, we prove in appendix \ref{constraintsproof} that this is the case: On a $t-\varphi$ reflection invariant Cauchy surface $\Sigma$ of the background spacetime, a solution to the linearized Einstein constraint equations always can be found for any given axisymmetric specifications of $\overline {\delta\bN}$, $\overline {\delta\bS}$, and $\overline {\delta\bJ}$. This completes the proof.
\end{proof}

We note the following simple corollary:

\begin{corollary}
A dynamical equilibrium configuration is in thermodynamic equilibrium if and only if $\delta M=0$ 
for all perturbations that satisfy the linearized Einstein constraint equations and for which $\delta S = \delta N = \delta J = 0$.
\end{corollary}

\begin{proof}
By the same type of argument as in the proof of the theorem, we have $\delta M=0$ 
for all perturbations that satisfy the linearized Einstein constraint equations and for which $\delta S = \delta N = \delta J = 0$ if and only if $\widetilde T$, $\widetilde\mu$, and $\Omega$ are uniform throughout the star.
\end{proof}

\section{Lagrangian Formulation of Perfect Fluids: Symplectic Structure, Phase Space, and Trivial Displacements} \label{Lagrangianform}

\subsection{Lagrangian and Symplectic Form}
\label{LagrangianformA}
As discussed in the Introduction, we wish to consider fluids whose local state is characterized by a particle number density, $n$, and entropy per particle, $s$. The energy density, $\rho$, is taken to be a prescribed function of these variables, $\rho = \rho(n,s)$, and the pressure is assumed to be given by \eqref{gibbsduhem}. The physical fields describing the fluid-gravitational system are thus $n$ and $s$, together with the fluid $4$-velocity, $u^a$, and the spacetime metric, $g_{ab}$. Equivalently, we may take the physical fields to be $(N_{abc}, s,g_{ab})$, since this contains exactly the same information as $(n, u^a, s,g_{ab})$, given the definition \eqref{Ndefn} and the normalization condition $u^a u_a = -1$.

Unfortunately, it is not possible to formulate an unconstrained Lagrangian description of a perfect fluid in terms of these physical fields \cite{SchutzSorkin}. Various Lagrangian formulations can be given in which potentials or other variables are taken to be the dynamical fields, from which the physical fields can then be obtained \cite{Iyer}. We will make use of a Lagrangian formulation described by Friedman \cite{Friedman} and others, in which a diffeomorphism, $\chi$, plays the role of the dynamical variable describing the fluid.

In this formulation, one introduces a {\em fiducial} manifold, $M'$, that is diffeomorphic to the spacetime manifold, $M$.  Then one chooses\footnote{\footnotesize In this section, we will assume for simplicity that $s'$ and $\bN'$ are everywhere nonvanishing, as this will make the discussion of the phase space less cumbersome. We are, of course, primarily interested in the case where $s'$ and $\bN'$ are nonvanishing only inside a worldtube (corresponding to a ``star''), but this can be straightforwardly dealt with by redefining $\chi$ below to be a map from the worldtube into $M$. \medskip}, on $M'$, a fixed scalar field $s'$ and a fixed 3-form field $\bN'$, satisfying
\be\label{conditions}
\ba
   d\bN' &= 0 ,\\
   d(s' \bN') &= 0.
\ea
\ee
The dynamical fields consist of a metric, $g_{ab}$, on $M$ and a
diffeomorphism $\chi:M'\to M$. We denote the dynamical fields collectively by 
\be
\phi=(g_{ab},\chi) \, .  
\ee
The
physical fluid fields on $M$ are defined by pushing forward with $\chi$:
\begin{align}
  \boldsymbol{N}&\equiv\chi_\ast \boldsymbol{N}'\,,\\
  s&\equiv\chi_\ast s'\,.
\end{align}
It follows from \eqref{conditions} that the field configurations allowed in this formalism automatically satisfy conservation of particle current, \eqref{eq:conservationofn}, and conservation of entropy along worldlines, \eqref{eq:conservationofs}. Any Einstein-fluid field configuration that satisfies these two conservation laws (but not necessarily the other field equations) can be constructed in this formalism by an appropriate choice of $\bN'$ and $s'$. However, once chosen, $\bN'$ and $s'$ are required to remain fixed. Consequently, the kinematically allowed field variations are those that correspond to ``moving fluid elements around"---without changing the number of particles or the entropy in any fluid element\footnote{\footnotesize In particular, the total number of particles and the total entropy cannot be varied.\medskip}---together with arbitrary changes to the metric. These kinematical restrictions are compatible with dynamical evolution, but they restrict the variations in initial conditions that one is allowed to consider.

It should be noted that there is a redundancy in this description of the fluid---in addition to the usual diffeomorphism redundancy of general relativity. Namely, two field configurations having the same spacetime metric, $\phi = \left(g_{ab}, \chi \right)$ and $\widetilde\phi = \left(g_{ab}, \widetilde\chi \right)$, are physically equivalent if they give rise to the same $\bN$ and $s$, i.e., if $\chi_\ast \boldsymbol{N}'=\widetilde\chi_\ast \boldsymbol{N}'$ and $\chi_\ast s'=\widetilde\chi_\ast s'$, or, equivalently, if $\bN'$ and $s'$ are unchanged under $\widetilde\chi^{-1}\circ\chi$ (which is a diffeomorphism on $M'$). We say that two such field configurations are {\em trivially related}.

The Lagrangian for the Einstein-fluid system is taken to be
\begin{equation} \label{L}
  \bL = \bL_{(g)} + \bL_{(m)} = \frac{1}{16\pi}R\bepsilon-\rho(n,s) \bepsilon \, .
\end{equation}
Here $\rho(n,s)$ is the function that specifies the energy density of the fluid under consideration in terms of $(n,s)$. However, in \eqref{L}, $\rho$ is to be viewed as a function of the dynamical variables $\phi=(g_{ab},\chi)$, from which $(n,s)$---and, hence $\rho$---can be computed, given the (fixed) specification of $\bN'$ and $s'$ on the fiducial manifold $M'$.
In order to apply the constructions\footnote{\footnotesize It was assumed in the constructions of section II that all of the dynamical fields are tensor fields. The diffeomorphism $\chi$ is not a tensor field, but by introducing local coordinates on $M'$, one may view $\chi^{-1}$ as a collection of $4$ scalar fields---namely, the maps from spacetime into each of the 4 coordinates on $M'$ (see the end of appendix \ref{phasespaceproof})---thereby allowing us to treat the dynamical fields within the framework of section II.\medskip} of section II, one must consider variations
about an arbitrary field configuration $\phi$.  To do this, introduce
a one-parameter family of dynamical fields, $\phi(\lambda)=\left(g_{ab}(\lambda), \chi_\lambda\right)$. The one-parameter family of diffeomorphisms $\chi_\lambda \circ \chi_0^{-1}:M\to M$ is generated to first order by a vector field $\xi^a$ known as a {\em Lagrangian displacement}. Thus, a first order perturbation is completely specified by a pair, $\delta \phi \equiv (\delta g_{ab}, \xi^a)$, consisting of a metric perturbation and a Lagrangian displacement. The first order variations of $\bN$ and $s$ are given by
\be
\delta \bN = - \lie_\xi \bN \, , \quad \quad \quad \delta s = - \lie_\xi s \, .
\label{lagfram}
\ee
Note that a general Einstein-fluid perturbation $(\delta g_{ab}, \delta \bN, \delta s)$ can be described within our Lagrangian framework if and only if there exists a vector field $\xi^a$ such that \eqref{lagfram} holds. This will be the case \cite{Friedman} if and only if (i) there is no variation of total particle number and entropy, $\delta N = \delta S = 0$, and (ii) $\delta s/|D_a s|$ is bounded (so, in particular, $\delta s = 0$ at any point where $\nabla_a s = 0$).
In accord with the remark at the end of the previous paragraph, a first order perturbation is said to be {\em trivial} if $\delta g_{ab}=0$, $\lie_\xi \bN$ = 0, and $\lie_\xi s=0$; i.e., if all of the physical variables are unchanged by the perturbation.

Following common terminology, a first order variation, $\delta \mathcal Q$, of an arbitrary tensor quantity $\mathcal Q$ on $M$ induced by $\delta \phi$ is called an {\em Eulerian perturbation}. More generally, the  $k^\text{th}$-order {\em Eulerian perturbation} of $\mathcal Q$ is given by
\begin{equation}
  \delta^k\mathcal{Q}\equiv\left.\frac{d^k}{d\lambda^k}\mathcal{Q}(\lambda)\right|_{\lambda=0}\,.
\end{equation}
However, for many purposes, it is convenient to pull back $\phi(\lambda)=\left(g_{ab}(\lambda), \chi_\lambda\right)$ by the spacetime diffeomorphism $\chi_\lambda \circ \chi_0^{-1}$ to obtain the gauge equivalent field configuration 
$\widehat{\phi}(\lambda)=\left((\chi_\lambda \circ \chi_0^{-1})^\ast g_{ab}(\lambda), \chi_0\right)$. This corresponds to expressing the $\phi(\lambda)$ in a gauge where the location of each fluid element in spacetime does not change with $\lambda$. We define the $k^\text{th}$-order {\em Lagrangian perturbation} of $\mathcal Q$ to be the $k^\text{th}$-order perturbation of $\mathcal Q$ in this gauge, i.e.,
\begin{equation}
  \Delta^k \mathcal{Q}\equiv\left.\frac{d^k}{d\lambda^k}\Bigl((\chi_\lambda \circ \chi_0^{-1})^\ast\mathcal{Q}(\lambda)\Bigr)\right|_{\lambda=0}.
\end{equation}
The Eulerian perturbations compare
$\mathcal{Q}(\lambda)$ and $\mathcal{Q}(0)$ at the same point $P\in M$, whereas the Lagrangian perturbations can be viewed as comparing $\mathcal{Q}(\lambda)$ and $\mathcal{Q}(0)$ at the same fluid element. It follows immediately that the Lagrangian perturbations of $\bN$ and $s$ vanish at all orders
\begin{align}
  \label{eq:Ds}\Delta^k s&=0\, ,\\
  \Delta^k \boldsymbol{N}&=0 \, .
\end{align}

Since, for any tensor quantity $\mathcal Q$, the first order Lagrangian perturbation, $\Delta \mathcal Q$, differs from the first order Eulerian perturbation, $\delta \mathcal Q$, by the action of an infinitesimal diffeomorphism generated by the Lagrangian displacement $\xi^a$, we have
\be \label{lagrangianpert}
\Delta \mathcal{Q} = \delta\mathcal{Q}+\lie_\xi\mathcal{Q} \, .
\ee
As noted above, we have $\Delta s = \Delta \bN = 0$, whereas by \eqref{lagrangianpert}, we have
\begin{equation}\label{eq:Dg}
  \Delta g_{ab} = \delta g_{ab} + 2\nabla_{(a}\xi_{b)} \, .
\end{equation}
The Lagrangian
perturbation of any other physical field can thus be expressed in
terms of $\Delta g_{ab}$ and background ($\lambda=0$) quantities. In particular, we obtain
\begin{align}
  \Delta\bepsilon &=\frac{1}{2}\bepsilon g^{ab}\Delta g_{ab}\, ,\\
  \label{Deltau} \Delta u^a &= \frac{1}{2}u^a u^b u^c \Delta g_{bc}\,,\\
  \Delta n &= -\frac{1}{2}nq^{ab}\Delta g_{ab}\,,
\end{align}
where 
\be
q_{ab}\equiv g_{ab}+u_au_b
\ee
is the projector orthogonal to $u^a$ and \eqref{Ndefn} together with the normalization condition, 
\be
N_{abc}N^{abc}=6 n^2\, ,
\ee
has been used.

Returning to the Lagrangian \eqref{L}, we see that variation of the matter part yields
\be
\label{Lm}
\ba
  \delta\bL_{(m)} &= -\delta\left(\rho \bepsilon\right) = -\Delta\left(\rho \bepsilon\right) + \lie_\xi\left(\rho \bepsilon\right)
  = -\Delta(\rho\bepsilon) + d(i_\xi \rho  \bepsilon)\\
  &= -\frac{\rho+p}{n}\bepsilon\Delta n - \frac{1}{2}\rho\bepsilon g^{ab}\Delta g_{ab} + d(i_\xi \rho  \bepsilon)\\
  &= \frac{1}{2}(\rho+p)\bepsilon q^{ab}\Delta g_{ab}-\frac{1}{2}\rho\bepsilon g^{ab}\Delta g_{ab} + d(i_\xi \rho  \bepsilon)\\
&= \frac{1}{2}T^{ab}(\delta g_{ab}+2\nabla_a\xi_b)\bepsilon +  d(i_\xi \rho  \bepsilon)\\
&= \frac{1}{2}T^{ab}\delta g_{ab}\bepsilon - \xi_b\nabla_aT^{ab}\bepsilon + \nabla_a(\xi^bT_{b}^{\phantom ba})\bepsilon+d(i_\xi \rho  \bepsilon) \\
&= \frac{1}{2}T^{ab}\delta g_{ab}\bepsilon - \xi_b\nabla_aT^{ab}\bepsilon + d(i_y \bepsilon) \, ,
\ea
\ee
where 
\be
y^a = \xi^b T_{b}^{\phantom ba} + \rho \xi^a \, .
\ee
Taking account of the variation of $\bL_{(g)}$ (see \eqref{Eeom}), we see that the equations of motion obtained from $\bL$ are
\begin{align}
-\frac{1}{16 \pi} G^{ab} + \frac{1}{2} T^{ab} &= 0\,, \\
- \nabla_aT^{ab} &= 0 \, .
\end{align}
Thus, $\bL$ yields the correct Einstein-fluid equations of motion.

From \eqref{Lm}, we also may read off the matter part of the symplectic potential current $\btheta^{(m)}$
\begin{equation}
  \theta^{(m)}_{abc}(\phi,\delta\phi) = \rho \xi^d \epsilon_{dabc} + \xi^d T_{d}^{\phantom de} \epsilon_{eabc}= (\rho+p)\xi^d q_d^{\phantom{d}e}\epsilon_{eabc} = \xi^d P_{dabc},
\end{equation}
where we have defined
\be \label{Pdefn}
P_{dabc}\equiv (\rho+p)q_d^{\phantom{d}e}\epsilon_{eabc}.
\ee
As explained in section II, in order to calculate the symplectic current $\bomega$ from $\btheta$ using \eqref{sympcurrent}, we need to choose an extension of $\delta_1 \phi = (\delta_1 g_{ab},\xi_1^a)$ and $\delta_2 \phi = (\delta_2 g_{ab},\xi_2^a)$ away from the field point $\phi$ at which we are calculating $\bomega$. We choose $\delta_1 g_{ab}$ and $\delta_2 g_{ab}$ to correspond to variations along a two parameter family of metrics, $g_{ab}\left(\lambda_1, \lambda_2\right)$, and we choose $\xi^a_1$ and $\xi^a_2$ to be fixed, i.e.,
\be
\delta_1 \xi_2^a = 0 \, , \quad \quad \delta_2 \xi_1^a = 0 \, .
\ee
With this choice, we have $\delta_2 \delta_1 g_{ab} = \delta_1 \delta_2 g_{ab}$ (partial derivatives with respect to $\lambda_1$ and $\lambda_2$ commute), whereas
\be \label{twodeltas}
\ba
\delta_1\delta_2 s - \delta_2\delta_1 s &= -\delta_1 \left( \xi_2^a \nabla_a s \right) + \delta_2 \left( \xi_1^a \nabla_a s \right) \\
&= - \xi_2^a \nabla_a \delta_1s  + \xi_1^a \nabla_a \delta_2s  \\
&=  \xi_2^a \nabla_a \left( \xi_1^b \nabla_b s\right)  -  \xi_1^a \nabla_a \left( \xi_2^b \nabla_b s\right) \\
&= [\xi_2,\xi_1]^b \nabla_b s \\
&= -\lie_{[\xi_1,\xi_2]} s,
\ea
\ee
and similarly,
\be
\delta_1 \delta_2 \bN - \delta_2 \delta_1 \bN =  - \lie_{[\xi_1,\xi_2]} \bN \, .
\ee
Thus, the perturbation $\delta_1\delta_2\phi - \delta_2\delta_1\phi$ is given (at $\phi$) by $\left(\delta g_{ab} = 0, \xi^a = [\xi_1,\xi_2]^a\right)$. 

We now have all we need to calculate the matter part of the symplectic current:
\be
\label{sympcur}
\ba
  \omega^{(m)}_{abc}(\phi;\delta_1\phi,\delta_2\phi) &= 
\delta_1 \theta^{(m)}_{abc}\left(\phi, \delta_2 \phi\right) - \delta_2 \theta^{(m)}_{abc}\left(\phi, \delta_1 \phi\right) - \theta^{(m)}_{abc}\bigl(\phi, \delta_1\delta_2 \phi - \delta_2\delta_1 \phi \bigr) \\
  &= \xi_2^d \delta_1 P_{dabc} - \xi_1^d \delta_2 P_{dabc} - [\xi_1,\xi_2]^d P_{dabc} \, .
\ea
\ee
Thus, the symplectic form \eqref{sympform2} is given by
\be \label{eq:symplecticform}
\ba
W[\phi;\delta_1\phi,\delta_2\phi]&=W^{(g)}[g_{ab};\delta_1g_{ab},\delta_2g_{ab}]+W^{(m)}[\phi;\delta_1\phi,\delta_2\phi]\\
&= \frac{1}{16\pi} \int_\Sigma \Bigl[ (\delta_2 h_{ij})(\delta_1 \pi^{ij}_{\phantom{ij}klm}) - (\delta_1 h_{ij})(\delta_2 \pi^{ij}_{\phantom{ij}klm}) \Bigr]\\
&\qquad\,+ \int_\Sigma \Bigl[ \xi_2^a \delta_1 P_{aklm} - \xi_1^a \delta_2 P_{aklm} - [\xi_1,\xi_2]^a P_{aklm} \Bigr] , 
\ea
\ee
where we have used the well known expression \cite{LeeWald,IyerWald,HollandsWald} for the contribution of $\bL_{(g)}$ to the symplectic current and
\be
\pi^{ij}_{\phantom{ij}klm} = \left(K^{ij}-h^{ij} K\right) \widehat\epsilon_{klm}
\ee
is the usual canonical momentum of general relativity, with $K_{ij}$ being the extrinsic curvature and $\widehat\bepsilon$ being the induced volume 3-form on $\Sigma$.

\subsection{Phase Space}
\label{LagrangianformB}

Following the prescription of Lee and Wald \cite{LeeWald}, phase space is constructed by factoring the space of all fields $(g_{ab}, \chi)$ on spacetime by the degeneracies of $W$, i.e., phase space is the space of equivalence classes of field configurations, where two field configurations are equivalent if they lie on an orbit of degeneracy directions of $W$. In vacuum general relativity, where $W$ is given by the first term on the right side of the second equality of \eqref{eq:symplecticform}, it follows immediately that phase space may be identified with the space of the fields (``initial data'')  $(h_{ij},\pi^{ij}_{\phantom{ab}klm})$ on $\Sigma$. However, on account of the presence of the commutator term in \eqref{eq:symplecticform}, it is not as straightforward to determine the phase space of the Einstein-fluid system.

To describe the phase space of the Einstein-fluid system, it is useful to introduce the $3$-manifold, $\Sigma'$, of ``fiducial flowlines'' on $M'$, defined as follows: Since $\bN'$ is a $3$-form on the fiducial $4$-manifold $M'$, there exists a nonvanishing vector field $U^{\prime a'}$ on $M'$---unique up to scaling at each point---such that $i_{U'} \bN' = 0$. The integral curves of $U^{\prime a'}$ are uniquely determined by $\bN'$ as unparameterized curves. We define $\Sigma'$ to be the manifold of orbits of $U^{\prime a'}$. We note that a Lagrangian formulation of the Einstein-fluid system---essentially equivalent to ours---can be given \cite{HawkingEllis} by taking the dynamical variable to be a smooth map from $M$ into $\Sigma'$ (rather than $M'$) with the additional requirement that the restriction of this map to any Cauchy surface be a diffeomorphism.

It is shown in appendix \ref{phasespaceproof} that the phase space of the Einstein-fluid system may be identified with the space of quantities $(h_{ij},  \pi^{ij}_{\phantom{ij}klm}, \psi, u^i)$ on $\Sigma$, where $u^i=h^{ia} u_a$ is the fluid 3-velocity, and $\psi$ is the diffeomorphism from $\Sigma'$ to $\Sigma$ obtained by intersecting with $\Sigma$ the images under $\chi$ of the fiducial flowlines. The statement that the phase space is given by $(h_{ij},  \pi^{ij}_{\phantom{ij}klm}, \psi, u^i)$ on $\Sigma$ is equivalent (by definition) to the statement that $\delta\phi$ is a degeneracy of $W$ if and only if $0=\delta h_{ij}=\delta \pi^{ij}_{\phantom{ij}klm} = \delta \psi = \delta u^i$ on $\Sigma$. Note that $\delta \psi = 0$ if and only if the Lagrangian displacement vector field $\xi^a$ on $\Sigma$ is parallel to the background $4$-velocity $u^a$, i.e., if and only if $q^a_{\phantom a b}\xi^b = 0$ on $\Sigma$.

Although $(h_{ij},  \pi^{ij}_{\phantom{ij}klm}, \psi, u^i)$ on $\Sigma$ provide coordinates on phase space, they are not ``canonically conjugate coordinates,'' as can be seen from the fact that the symplectic product of two pure $\psi$ perturbations does not vanish in general. For the purpose of introducing a Hilbert space structure on perturbations, it is useful to introduce canonically conjugate coordinates $(q^\alpha, p_\alpha)$ such that $W$ takes the form
\be
W[\phi; \delta_1 \phi, \delta_2 \phi] =  \int_\Sigma \sum_\alpha (\delta_2 q^\alpha \cdot \delta_1 p_\alpha - \delta_1 q^\alpha \cdot \delta_2 p_\alpha) \, ,
\label{ccc}
\ee
where each $q^\alpha$ is a tensor field on $\Sigma$ and each $p_\alpha$ is a tensor density on $\Sigma$ dual to $q^\alpha$. In appendix \ref{phasespaceproof} we show how to obtain such canonically conjugate variables by representing the dynamical diffeomorphism $\chi$ as a collection of scalar fields. As seen in appendix \ref{phasespaceproof}, $q^\alpha$ consists of the perturbation to the spatial metric, $\delta h_{ij}$, together with $3$ scalar fields representing the fluid perturbation, but the explicit form of \eqref{ccc} is not needed here.

Using such canonically conjugate coordinates, we can define a Hilbert space structure $\mathcal K$ on perturbations by introducing the $L^2$ inner product\footnote{\label{vol}\footnotesize If $q^\alpha$ is a tensor and $p_\alpha$ is a tensor density with dual indices as assumed above, then no volume element need be specified in \eqref{ccc}. However, a volume element must be specified in \eqref{Kdef}. If we take the volume element in \eqref{Kdef} to be a fixed volume element on $\Sigma$, then the term $|q^\alpha|^2$ should be multiplied by $h^{1/2}$ and the term $|p_\alpha|^2$ should be multiplied by $h^{-1/2}$ where $h$ denotes the determinant of the background spatial metric $h_{ab}$ on $\Sigma$ with respect to the fixed volume element. For notational simplicity, we have ignored these factors, since, for any fixed background, we may assume that $h=1$.\medskip}
\be
\langle \delta_1 \phi, \delta_2 \phi \rangle =  \int_\Sigma \sum_\alpha (\delta_1 q^\alpha \cdot \delta_2 q^\alpha + \delta_1 p_\alpha \cdot \delta_2 p_\alpha) \, ,
\label{Kdef}
\ee
where ``$\cdot$'' now denotes contraction of all tensor indices after using the background metric $h_{ab}$ on $\Sigma$ to raise and lower indices. Thus, the elements of $\mathcal K$ are the square integrable tensor fields $(q^\alpha, p_\alpha)$ on $\Sigma$. Note that perturbations for which $\delta M \neq 0$ fall off too slowly to be square integrable, but $\mathcal K$ contains all perturbations of interest for which $\delta M = 0$.

By inspection of \eqref{ccc} and \eqref{Kdef}, it can be seen that $W$ is a bounded quadratic form on $\mathcal K$ and thus corresponds to a bounded linear map $\widehat{W} : \mathcal K \to \mathcal K$ such that
\be
W\left[\phi; \delta_1\phi, \delta_2 \phi \right] = \langle \delta_1\phi, \widehat W \delta_2 \phi \rangle.
\ee
It is not difficult to see that
\be
\widehat{W} (q^\alpha, p_\alpha) = (-p_\alpha, q^\alpha) \, ,
\label{Wmap2}
\ee
where it is understood that any tensor indices on $(q^\alpha, p_\alpha)$ are converted to the corresponding dual indices on the right side of this equation via raising and lowering with $h^{ab}$ and $h_{ab}$ and we have assumed $h=1$ (see footnote \ref{vol}). It follows immediately from \eqref{Wmap2} that $\widehat{W}^2 = - I$ and $\widehat{W}^\dagger = - \widehat{W}$, so, in particular, $\widehat{W}$ is an orthogonal map. 

Let $\mathcal S$ be any subspace of $\mathcal K$. We define the {\em symplectic complement}, ${\mathcal S}^{\perp_S}$, of $\mathcal S$ by
\be
{\mathcal S}^{\perp_S} = \big\{v \in \mathcal K \big| \langle v, \widehat{W} u \rangle = 0 \,\,\, \forall u \in \mathcal S \big\} \, .
\ee
Clearly, we have ${\mathcal S}^{\perp_S} = (\widehat{W}[\mathcal S])^\perp$, where $\widehat{W}[\mathcal S]$ denotes the image of $\mathcal S$ under $\widehat{W}$ and ``$\perp$'' denotes the ordinary orthogonal complement in $\mathcal K$. Since $\widehat{W}$ is orthogonal, we have $(\widehat{W}[\mathcal S])^\perp = \widehat{W}[\mathcal S^\perp]$, and since $\widehat{W}^2 = - I$, we have
\be
({\mathcal S}^{\perp_S})^{\perp_S} = (\widehat{W}[\mathcal S^\perp])^{\perp_S} 
= \widehat{W}^2 [(\mathcal S^\perp)^\perp] = (\mathcal S^\perp)^\perp = \overline{\mathcal S} \, ,
\label{doubsym}
\ee
where the bar denotes the closure in $\mathcal K$. Thus, the double symplectic complement of any subspace is its closure.

Now let $\phi$ satisfy the equations of motion and let $X^a$ be smooth and of compact support. By \eqref{ham0}, we have for all $\delta \phi  \in \mathcal K$
\be
\langle \delta \phi, \widehat W \lie_X \phi \rangle = \int_\Sigma X^a \delta \bC_a \, .
\ee
By definition, the right side of this equation vanishes if and only if $\delta \phi$ is a weak solution of the constraint equations, $\delta \bC_a = 0$. Thus, if we take $\mathcal G$ to be the subspace of $\mathcal K$ spanned by perturbations of the form $\lie_X \phi$, we see that ${\mathcal G}^{\perp_S}$ is precisely the subspace, $\mathcal C$, of weak solutions to the constraints. Furthermore, by the general argument of the previous paragraph, we have 
${\mathcal C}^{\perp_S} = \overline{\mathcal G}$. Another way of saying this is that if we restrict the action of the original quadratic form $W$ to $\mathcal C \times \mathcal C$, it becomes degenerate precisely on (the closure of) the gauge transformations $\lie_X \phi$.

\subsection{Trivial Displacements}
\label{LagrangianformC}

We will see in the next section that it will be important important to determine the symplectic complement, $\mathcal V$, within $\mathcal C$ of the subspace of phase space perturbations corresponding to field variations of the form $(\delta g_{ab} = 0, \eta^a)$, where $\eta^a$ is 
a trivial displacement, i.e., 
\be\ba
  0 &= \delta s = -\lie_\eta s \,, \\
  0 &= \delta \bN = -\lie_\eta \bN \, .
\label{trivials}
\ea\ee
We first find the general form of a trivial displacement. Since $u^a N_{abc} = 0$, any vector field $\eta^a$ inside the star can be uniquely decomposed as
\be
\eta^a = f u^a + \frac{1}{n^2} N^{abc} H_{bc} \, ,
\ee
where $f$ is an arbitrary function and $H_{ab}$ is an arbitrary $2$-form satisfying $u^a H_{ab} = 0$. 
Since $d \bN = 0$, the necessary and sufficient condition to satisfy the second equality in \eqref{trivials} is
\be
0 = \lie_\eta \bN = d (i_\eta \bN) = 2 d {\boldsymbol H}\, ,
\label{dH}
\ee
where we have used
\be
\frac{1}{n^2} N^{abc} N_{ade} = 2 {q^{[b}}_d {q^{c]}}_e 
\ee
to calculate $i_\eta \bN$. It follows immediately that $\lie_u {\boldsymbol H} = i_u d{\boldsymbol H} + d (i_u {\boldsymbol H}) = 0$, so $\boldsymbol H$ may be viewed as a $2$-form on the manifold of orbits of $u^a$. Assuming that the star is simply connected, \eqref{dH} then yields
\be
\boldsymbol H = d \boldsymbol Z \, ,
\ee
where $\boldsymbol Z$ is an arbitrary $1$-form on the manifold of $u^a$-orbits or, equivalently, $\boldsymbol Z$ is a $1$-form on spacetime satisfying $i_u \boldsymbol Z = 0$ and $\lie_u \boldsymbol Z = 0$. Thus, the necessary and sufficient condition for $\eta^a$ to satisfy $\lie_\eta \bN = 0$ is that it be of the form
\be
\eta^a = f u^a + \frac{1}{n^2} N^{abc} \nabla_b Z_c \, ,
\label{z1}
\ee
where $Z_a$ satisfies
\be
u^a Z_a = 0 \, , \quad \lie_u Z_a = 0 \, .
\label{z2}
\ee
Since $u^a \nabla_a s = 0$, the necessary and sufficient condition for $\eta^a$ to also satisfy $\eta^a \nabla_a s = 0$ is
\be
\nabla_{[a} s \nabla_b Z_{c]} = 0 \, .
\label{z3}
\ee
Eqs. \eqref{z1}--\eqref{z3} are necessary and sufficient for $\eta^a$ to be a trivial displacement.

It should be noted that if $\nabla_a s \neq 0$ and the surfaces of constant $s$ are simply connected (i.e., spheres), then it is possible to show further that
\be
\eta^a = f u^a + \frac{1}{n^2} N^{abc} (\nabla_b s) (\nabla_c F) \, ,
\label{frtriv1}
\ee
where $\lie_u F = 0$, which is the form given in \cite{Friedman} for the case where $\nabla_a s \neq 0$. However, if the surfaces of constant $s$ are not simply connected (i.e., tori), then (a small class of) additional trivials are also allowed. Similarly, if $\nabla_a s = 0$ in an open region, then writing $Z_a$ as a sum of terms of the form $F \nabla_a G$, it can be seen from \eqref{z1} and \eqref{z2} that $\eta^a$ can be written as a sum of terms of the form \cite{Friedman}
\be
\eta^a = f u^a + \frac{1}{n^2} N^{abc} (\nabla_b F_1) (\nabla_c F_2) \, ,
\label{frtriv2}
\ee
where $\lie_u F_1 = \lie_u F_2 = 0$. However, in order to avoid dealing with these different special cases, we will use the form \eqref{z1}--\eqref{z3}, which is valid in all cases.

We now compute the symplectic product of a trivial perturbation, $(\delta g_{ab} = 0, \eta^a)$, with an arbitrary perturbation. It is not difficult to see that all of the flowline trivials, $\eta^a = f u^a$ for any $f$, are degeneracies of $W$ since they have $0=\delta h_{ij}=\delta \pi^{ij}_{\phantom{ij}klm} = q^a_{\phantom a b}\xi^b = \delta u^i$ on $\Sigma$. Thus, these trivial perturbations are not represented in phase space, i.e., they are ``factored out'' by our above construction of phase space. However, all of the nonvanishing trivials $\widetilde\eta^a$ of the form 
\be
\widetilde{\eta}^a = \frac{1}{n^2} N^{abc} \nabla_b Z_c 
\label{ntt}
\ee
with $Z_a$ satisfying \eqref{z2} and \eqref{z3} are {\em not} degeneracies of $W$. Indeed, for an arbitrary $\delta \phi$, by working in a gauge where the Lagrangian displacement is zero (which we can do since the symplectic product is gauge invariant), we have
\be
\ba \label{symportho}
W[\phi;\delta\phi,(0,\widetilde\eta^a)] 
&= W[\phi;(\Delta g_{ab},0),(0,\widetilde\eta^a)] \\
&= \int_\Sigma \widetilde\eta^a \Delta P_{abcd} \\
&= \int_\Sigma \frac{1}{n^2} N^{aef} \nabla_e Z_f \Delta P_{abcd}\\
&= \int_\Sigma \nabla_e Z_f \Delta\left( \frac{1}{n^2} N^{aef}P_{abcd}\right) \\
&= 6\int_\Sigma \nabla_{[b} Z_c  \Delta\left( \frac{\rho+p}{n} u_{d]} \right) \\
&= 6\int_\Sigma  Z_{[b}  \Delta\left(\nabla_c \frac{\rho+p}{n} u_{d]} \right)\\
&= \int_\Sigma  \boldsymbol Z \wedge \Delta d\left(\frac{\rho+p}{n} \boldsymbol u\right).
\ea
\ee
The 2-form $d\left[(\rho+p) \boldsymbol u /n\right]$ is known as the {\em vorticity}, so we see that a sufficient condition for symplectic orthogonality to the trivials is vanishing Lagrangian change of the vorticity\footnote{\label{trivfoot}\footnotesize A necessary (but not sufficient, unless the level surfaces of $s$ in the background are spheres) condition for symplectic orthogonality to the trivials is vanishing Lagrangian change in the quantity $ds\wedge d\left(\frac{\rho+p}{n} \boldsymbol u\right)$, known as the {\em circulation}, as can be seen by considering the particular trivials of the form $\boldsymbol Z = F ds$ in \eqref{symportho}. In the case of nonaxisymmetric perturbations, Friedman \cite{Friedman} shows that for certain backgrounds the condition of vanishing Lagrangian change in circulation is not a physical restriction---in the sense that for any perturbation $\delta\phi$, one can always find a trivial perturbation to add to $\delta\phi$ such that the sum has zero Lagrangian change in circulation. However in the axisymmetric case the corresponding condition, $\Delta j = 0$, is a physical restriction.}. This condition is necessary in open regions where $\nabla_a s=0$.

Consider, now, the case of axisymmetric trivials. Then it is easy to verify directly from \eqref{trivials} that
\be
\eta^a = f \varphi^a
\ee
is a trivial displacement\footnote{\footnotesize In fact, it can be seen from \eqref{frtriv1} that in the case where the level surfaces of $s$ in the background are spheres, any axisymmetric trivial displacement can be written in the form $\eta^a = f_1 u^a + f_2 \varphi^a$ with $\lie_u f_2 = 0$.} for any axisymmetric function $f$ satisfying $\lie_u f=0$. We now show that the time derivative, $\lie_t \eta^a$ of {\em any} axisymmetric trivial is a trivial of this form, up to the addition of a flowline trivial. To see this, we take the Lie derivative of \eqref{ntt} and use the circular flow condition \eqref{cirflow} of the background to write

\newpage

\be \label{liettrivial}
\ba
\lie_t\widetilde\eta^a &=\lie_t \left(  \frac{1}{n^2} N^{abc} \nabla_b Z_c \right) \\
&=   \frac{1}{n^2} N^{abc} \left[ \nabla_b \left(\lie_{|v|u} Z_c \right) - \lie_{\Omega\varphi}\left( \nabla_b Z_c \right) \right] \\
&=  -2\frac{1}{n^2} N^{abc}  \left(\nabla_b \Omega\right) \left(\varphi^d \nabla_{[d} Z_{c]}\right),
\ea
\ee
where in the last equality we used axisymmetry of $\widetilde\eta^a$ and the properties \eqref{z2} of $Z_a$. Since both $\nabla_a \Omega$ and $\varphi^b \nabla_{[b} Z_{a]}$ vanish when contracted with $\varphi^a$, it follows from \eqref{liettrivial} that $\lie_t\widetilde\eta^a$ must be proportional to $\varphi^a$, which establishes our claim.

Finally, in parallel to \eqref{symportho}, the symplectic product of an arbitrary axisymmetric perturbation with an axisymmetric trivial of the form $\eta^a = f \varphi^a$ is
\be
\ba \label{symportho2}
W[\phi;\delta\phi,(0,f \varphi^a)] 
&= W[\phi;(\Delta g_{ab},0),(0,f\varphi^a)] \\
&= \int_\Sigma f \varphi^a \Delta P_{abcd} \\
&= \int_\Sigma f \varphi^a \Delta \left[ \frac{\rho+p}{n}\left(n \epsilon_{abcd} + u_a N_{bcd}\right)\right] \\
&= \int_\Sigma f N_{bcd} \Delta \left( \frac{\rho+p}{n} \varphi^a u_a \right),  \\
\ea
\ee
where we have chosen $\Sigma$ to be axisymmetric so that the pullback of $\varphi^a \epsilon_{abcd}$ vanishes. Thus, in the axisymmetric case, the necessary and sufficient condition for symplectic orthogonality to the trivials of the form $\eta^a = f\varphi^a$ is 
\be
\Delta j = 0 \, ,
\ee
where
\be
j \equiv \frac{\rho+p}{n}\varphi^a u_a
\ee
has the interpretation of being the ``angular momentum per particle''. 

\section{Canonical Energy and Dynamic Stability} \label{canonicalenergy}

The {\em canonical energy} $\mathcal E$ is a bilinear form on the space of solutions to the perturbation equations, defined by
\be
\mathcal E(\delta_1 \phi, \delta_2 \phi) = W[\phi; \delta_1 \phi,  \lie_t\delta_2 \phi],
\label{canen}
\ee
where $t^a$ is the timelike Killing vector field of the stationary background. More precisely, we define $\mathcal E$ to be the quadratic form \eqref{canen} on the Hilbert space $\mathcal C \subset \mathcal K$ defined at the end of subsection \ref{LagrangianformB}, with domain taken to be the smooth elements $\delta \phi \in \mathcal C$ with suitable decay properties at infinity\footnote{\footnotesize More precisely, the domain is ${\mathcal U} \cap {\mathcal C}$ where $\mathcal U$ is the intersection of weighted Sobolev spaces analogous to those defined in \cite{HollandsWald}. This domain can be shown to be dense in $\mathcal C$ by the type of argument given in proposition 5 of \cite{HollandsWald}.\medskip}.  An explicit formula for $\mathcal E$ is given in \cite{Friedman}, and an expression for $\mathcal E$ in terms of second order variations will be obtained in the next section.

Although the definition of $\mathcal E$ is asymmetric in $\delta_1 \phi$ and  $\delta_2 \phi$, it is, in fact, symmetric in its arguments,
\be
\mathcal E(\delta_1 \phi, \delta_2 \phi) = \mathcal E(\delta_2 \phi, \delta_1 \phi).
\label{Esym}
\ee
To prove this, we note that \eqref{sympcur} expresses the symplectic current $\bomega$ in terms of the perturbations $(\delta_1 \phi, \delta_2 \phi)$ and the background physical quantities $g_{ab}$, $\rho$, $p$, and $u^a$. Since $\lie_t$ applied to the background physical quantities vanishes, we have
\be\label{Esym2}
\lie_t \bomega(\phi; \delta_1 \phi, \delta_2 \phi) = \bomega(\phi; \lie_t \delta_1 \phi, \delta_2 \phi) 
+ \bomega(\phi; \delta_1 \phi, \lie_t \delta_2 \phi)\, .
\ee
By the standard Lie derivative identity for forms, we have
\be
\lie_t \bomega = i_t d \bomega + d (i_t \bomega) = d(i_t \bomega) \,,
\ee
where the fact that $\bomega$ is closed (see \eqref{sympcons}) was used in the last equality. Integration of \eqref{Esym2} over a Cauchy surface $\Sigma$ then yields \eqref{Esym}.

In addition to its symmetry, $\mathcal E$ satisfies the following important properties: (i) $\mathcal E$ is conserved, i.e., $\mathcal E$ is independent of the choice of $\Sigma$. This follows immediately from the conservation of $W$ on solutions (see \eqref{sympcons}), given that if $\delta \phi$ is a solution, then so is $\lie_t \delta \phi$. (ii) $\mathcal E$ is gauge invariant for gauge transformations of compact support. This follows immediately from \eqref{giW}. (iii) $\mathcal E(\delta \phi, \delta \phi)$ has a positive net flux at null infinity if the perturbation is asymptotically stationary at late times. This has been shown in \cite{Friedman} and \cite{HollandsWald}.

There is one further property that we need $\mathcal E$ to satisfy in order to use the positivity of $\mathcal E$ as a necessary and sufficient condition for stability: We want $\mathcal E$ to be degenerate precisely on the linearized solutions $\delta \phi$ that are physically stationary. (Here, $\mathcal E$ is said to be {\em degenerate} on $\delta \phi$ if ${\mathcal E}(\delta \phi, \delta \phi') = 0$ for all $\delta \phi'$ in the domain of $\mathcal E$.) Below, we will define what we mean by a ``physically stationary linearized solution,'' and we will then explain why this degeneracy property is needed in order to use positivity of $\mathcal E$ as a criterion for dynamic stability. Unfortunately, we will then find that $\mathcal E$ is {\em not} degenerate on all physically stationary solutions. The cure for this difficulty will be to restrict the subspace of solutions on which $\mathcal E$ is defined so as to make it degenerate on the physically stationary solutions. As a consequence, we can only directly test dynamic stability on a restricted subspace of perturbations. Nevertheless, we will then show that, in the axisymmetric case, mode stability on this restricted subspace implies mode stability for general perturbations, including perturbations that cannot be obtained in the Lagrangian framework.

A smooth linearized solution $\delta \phi = (\delta g_{ab}, \xi^a)$ is said to be {\em physically stationary} if the physical fields $\delta g_{ab}$, $\delta \bN$, and $\delta s$ can be made stationary by a gauge transformation, i.e., if there exists a smooth vector field $X^a$, which is an asymptotic symmetry near infinity, such that
\be \label{gcond}
0 = \lie_t [\delta g_{ab} + \lie_X g_{ab}] \,,
\ee
\be 
0 = \lie_t [\delta \bN + \lie_X \bN] = \lie_t [-\lie_\xi \bN + \lie_X \bN] = -\lie_{[t,\xi-X]} \bN \,,
\label{Ncond}
\ee
\be
0 = \lie_t [\delta s + \lie_X s] = \lie_t [-\lie_\xi s + \lie_X s] = -\lie_{[t,\xi-X]} s \, .
\label{scond}
\ee
Equations \eqref{Ncond} and \eqref{scond} are equivalent to the statement that the perturbation \mbox{$(0,[t,\xi-X]^a)$} is trivial, i.e., 
\be \label{xicond}
\lie_t [ \xi^a - X^a] = \text{trivial displacement}.
\ee
Thus, $\delta \phi = (\delta g_{ab}, \xi^a)$ is physically stationary if and only if there exists a smooth vector field $X^a$, which is an asymptotic symmetry near infinity, such that
\be \label{physstat}
\lie_t \delta\phi= \left(-\lie_{[t,X]} g_{ab}, [t,X]^a\right)  + \text{trivial} \, .
\ee

Now, if it were true that $\mathcal E(\delta \phi, \delta \phi) > 0$ for all linearized solutions $\delta \phi$, then $\mathcal E$ would provide a positive definite conserved norm, thereby implying mode stability (see the discussion given in the Introduction). Since physically stationary perturbations are obviously physically stable, we would also have mode stability if we merely had $\mathcal E(\delta \phi, \delta \phi) \geq 0$ for all linearized solutions $\delta \phi$, but with equality holding only for physically stationary perturbations. In other words, it does no harm to the argument for dynamic stability if we merely have $\mathcal E(\delta \phi, \delta \phi) \geq 0$ provided that $\mathcal E$ is degenerate {\em only} on physically stationary solutions. On the other hand, we need $\mathcal E$ to be degenerate on {\em all} physically stationary solutions in order to argue for instability in the alternative case where $\mathcal E(\delta \phi, \delta \phi) < 0$ for some linearized solution $\delta \phi$. Specifically, we need degeneracy of $\mathcal E$ on physically stationary solutions in order to obtain a contradiction with such a $\delta \phi$ asymptotically approaching a physically stationary solution at late retarded times: If $\delta \phi$ asymptotically approached a physically stationary solution, the above positive flux result would imply that $\mathcal E$ could only become more negative at late times, whereas the degeneracy of $\mathcal E$ on physically stationary solutions would imply ${\mathcal E} \to 0$, thus yielding a contradiction. Thus, we need $\mathcal E$ to be degenerate precisely on the physically stationary solutions in order to use positivity of $\mathcal E$ as a criterion for both stability and instability, i.e., to be able to prove that (i) non-negativity of $\mathcal E$ implies mode stability and (ii) failure of non-negativity implies the existence of solutions that cannot asymptote to a physically stationary final state.

What are the degeneracies of $\mathcal E$? Since $\mathcal E(\delta \phi', \delta \phi) = W(\phi; \delta \phi', 
\lie_t\delta \phi)$, it follows that $\delta \phi$ is a degeneracy of $\mathcal E$ if and only if $\lie_t \delta \phi$ is a degeneracy of $W$. As discussed at the end of subsection \ref{LagrangianformB}, when restricted to $\mathcal C$, $W$ is degenerate precisely on the gauge transformations that go to zero at infinity. Thus, the degeneracies of $\mathcal E$ are precisely the $\delta \phi$ in the domain of $\mE$ such that
\be
\lie_t \delta \phi = \left(\lie_Y g_{ab}, -Y^a\right) \,,
\label{Wdeg}
\ee
where $Y^a$ is smooth and goes to zero at infinity. Comparison of eqs. \eqref{physstat} and \eqref{Wdeg} shows that the degeneracies of $\mathcal E$ are a proper subset of the physically stationary solutions. Thus, although $\mE$ satisfies the desired property of being degenerate only on physically stationary solutions, it fails to be degenerate on all physically stationary solutions.

A cure for this difficulty is to restrict $\mathcal E$ to a smaller space, so as to make it degenerate on all physically stationary solutions. If $\delta_\text{ps}\phi$ is a physically stationary perturbation, we have, from \eqref{physstat},
\be
\ba
\mathcal E (\delta \phi, \delta_\text{ps}\phi ) &= W[\delta \phi, \lie_t 
\delta_\text{ps} \phi] \\
&= -W\left[\delta\phi, \left(\lie_{[t,X]} g, -[t,X]\right) \right] + W[\delta
\phi, \text{trivial}] \, .
\ea
\ee
Now, for a general asymptotic symmetry $X^a$, the commutator $[t,X]^a$ is, at most, an asymptotic translation (as occurs when $X^a$ is an asymptotic boost). Therefore, in order to ensure that the first term vanishes, we must restrict $\delta \phi$ so that $\delta P_i = 0$, where $\delta P_i$ denotes the ADM linear momentum (see \eqref{ham1} and \eqref{ham5}). This is an innocuous restriction on perturbations, since we can achieve this by addition of the action of an infinitesimal Lorentz boost on the background solution, so $\delta P_i = 0$ does not impose a physical restriction on the perturbations being considered. On the other hand, in order to ensure that the second term vanishes, we must restrict $\delta \phi$ so that
 \be \label{Tcond1}
W[\delta\phi, \text{trivial}] = 0
\ee
for all trivials. 

Let $\mathcal V$ be the Hilbert subspace of $\mathcal C$ composed of perturbations that have $\delta P_i = 0$ and are symplectically orthogonal to all trivials. Then $\mathcal V$ is the symplectic complement in the Hilbert space $\mathcal K$ of the subspace, $\mathcal W$, of perturbations consisting of trivials together with gauge transformations generated by vector fields $X^a$ that approach asymptotic translations at infinity. Since the double symplectic complement of $\mathcal W$ in $\mathcal K$ is simply the closure, $\overline{\mathcal W}$, of $\mathcal W$ in $\mathcal K$ (see \eqref{doubsym}), when restricted to $\mathcal V$, the degeneracies of $W$ are precisely the elements of $\overline{\mathcal W} \cap {\mathcal V}$. Furthermore, by arguments similar to given in \cite{HollandsWald} (see remark 2 of section 4 of that reference), the smooth elements of $\overline{\mathcal W}$ lie in $\mathcal W$. It follows that, when restricted to $\mathcal V$, $\mathcal E$ is degenerate precisely on the physically stationary solutions.

Putting together all of the above results and arguments, we have the following theorem:

\begin{theorem}\label{dynstabthm}
Let $\mathcal V \subset \mathcal C$ be the space of linearized solutions within the Lagrangian framework that are symplectically orthogonal to all trivial perturbations and satisfy $\delta P_i = 0$. If $\mathcal E$ is non-negative on this subspace, then one has stability on this subspace of perturbations in the sense that there do not exist any exponentially growing modes lying in this subspace. Conversely, if ${\mathcal E} (\delta \phi, \delta \phi) < 0$ for some $\delta \phi \in \mathcal V$, then one has instability in the sense that such a $\delta \phi$ cannot approach a physically stationary solution at asymptotically late times. 
\end{theorem}

We note that Friedman \cite{Friedman} has shown that if $\Omega$ is not identically zero, there exist perturbations in $\mathcal V$ of sufficiently high angular quantum number $m$ such that $\mathcal E < 0$, thus establishing that all rotating stars are dynamically unstable (the CFS instability) in the sense of this theorem. Furthermore, for slowly rotating stars, all ``$r$-modes'' with $m\geq2$ have $\mathcal E<0$ and thus are unstable \cite{FriedmanMorsink, Andersson}. For slowly rotating stars, the growth timescale of the unstable modes will be very long (see footnote \ref{stabterm}), but the instability may occur on dynamically relevant timescales for rapidly rotating compact stars.

As previously mentioned in the Introduction and footnote \ref{trivfoot}, Friedman \cite{Friedman} has shown that for non-axisymmetric perturbations, restriction to $\mV$ does not impose a (significant) physical restriction on perturbations, i.e., for suitable background stars, any nonaxisymmetric perturbation can be written as the sum of a trivial perturbation and a perturbation in $\mV$. However, in the axisymmetric case, restriction to $\mV$ does impose a physical restriction on perturbations. In particular, as shown in subsection \ref{LagrangianformC}, symplectic orthogonality to trivials of the form $f \phi^a$ requires $\Delta j = 0$, which is a significant physical restriction. It is worth noting that by eqs. \eqref{firstlaw2} and \eqref{Jeq2} expressed in a gauge where $\delta = \Delta$, the condition $\Delta j = 0$ implies $\delta J = \delta M = 0$, so all perturbations in $\mV$ satisfy\footnote{\footnotesize The Hilbert space $\mathcal K$ excludes perturbations with $\delta M \neq 0$ in any case because of the failure of square integrability. \medskip} $\delta J = \delta M = 0$. It is interesting that the same condition $\delta J = \delta M = 0$ in a space directly analogous to $\mV$ also arose in the black hole stability analysis of Hollands and Wald \cite{HollandsWald}, but for completely different reasons (involving the horizon Killing field). 

On account of the physical restrictions associated with considering only perturbations in $\mV$, Theorem \ref{dynstabthm} is of rather limited utility as it stands for determining the dynamic stability of a star with respect to axisymmetric perturbations, since it gives a stability criterion only for perturbations in $\mV$. Fortunately, in the axisymmetric case, these restrictions can be removed: Mode stability for perturbations in $\mathcal V$ implies mode stability for {\em all} perturbations, including those that cannot be described within the Lagrangian displacement framework. This result is a direct consequence of the following lemma:

\begin{lemma}
Let $\delta \mathcal Q = (\delta N_{abc}, \delta s, \delta g_{ab})$ be an axisymmetric solution to the linearized Einstein-fluid equations (not necessarily arising in the Lagrangian displacement framework). Then there exists a vector field $\xi^a$ such that
\be \label{Tcondition3}
\ba
\lie_t \delta N_{abc} &= -\lie_\xi N_{abc}\,, \\
\lie_t \delta s &= -\lie_\xi s\,, \\
\lie_t \delta j &= -\lie_\xi j \, .
\ea
\ee
Thus, $\lie_t \delta \mathcal Q$ can be represented in the Lagrangian displacement framework and has $\Delta j = 0$. Furthermore $\lie^2_t \delta \mathcal Q \in \mathcal V$.
\end{lemma}

\begin{proof}
Let
\be
\xi^a = |v| \delta u^a + \beta \varphi^a,
\ee
where $v^a = t^a + \Omega \varphi^a$ and $\beta$ is any axisymmetric scalar that satisfies 
\be
u^a \nabla_a \beta = (\delta u^a)\nabla_a \Omega \, .
\ee
The perturbed conservation of entropy equation yields
\be
\ba
0 &= \delta (u^a \nabla_a s ) = u^a \nabla_a \delta s + (\delta u^a)\nabla_a s \\
&= \frac{1}{|v|} (t^a + \Omega \varphi^a) \nabla_a \delta s +  (\delta u^a)\nabla_a s = \frac{1}{|v|} t^a \nabla_a \delta s +  (\delta u^a)\nabla_a s
\ea
\ee
where we have used axisymmetry of the perturbation in the last step. Thus, we have
\be
\lie_t \delta s = -|v|(\delta u^a) \nabla_a s =- \lie_\xi s.
\ee
An identical calculation using the perturbed conservation of angular momentum equation, $\delta(u^a \nabla_a j)=0$, shows
\be
\lie_t \delta j = - \lie_\xi j.
\ee
Finally, the perturbed conservation of particle number yields $\delta( d \boldsymbol{N}) = d(\delta \boldsymbol N) = 0$, so 
\be
\ba
\lie_t \delta \boldsymbol{N} &= t \cdot d(\delta \boldsymbol{N}) + d (t \cdot \delta \boldsymbol N) 
= d\left[ (|v|u-\Omega \varphi)\cdot \delta \boldsymbol N \right] \\
&= d\left[ |v|\delta(u\cdot \boldsymbol N)-|v|(\delta u)\cdot \boldsymbol N - \Omega \varphi\cdot \delta \boldsymbol N \right]\\
&= d\left[- \xi\cdot \boldsymbol N +\beta \varphi \cdot \boldsymbol N - \Omega \varphi \cdot \delta \boldsymbol N  \right]\\
&= -\lie_\xi \boldsymbol N + d\left[\varphi \cdot \left(\beta \boldsymbol N - \Omega \delta \boldsymbol N  \right)  \right]\\
&= -\lie_\xi \boldsymbol N + \lie_\varphi\left(\beta \boldsymbol N - \Omega \delta \boldsymbol N  \right)-\varphi \cdot d\left(\beta \boldsymbol N - \Omega \delta \boldsymbol N  \right)\\
&=-\lie_\xi \boldsymbol N -\varphi \cdot d\left(\beta \boldsymbol N - \Omega \delta \boldsymbol N  \right).
\ea
\ee
But $d\left(\beta \boldsymbol N - \Omega \delta \boldsymbol N  \right)$ is a 4-form, so 
\be
\ba
&\varphi \cdot d\left(\beta \boldsymbol N - \Omega \delta \boldsymbol N  \right)=0 \\
\Longleftrightarrow\quad &d\left(\beta \boldsymbol N - \Omega \delta \boldsymbol N  \right)=0\\
\Longleftrightarrow\quad & \boldsymbol{N}\wedge d\beta - \delta\boldsymbol{N}\wedge d\Omega = 0\\
\Longleftrightarrow\quad & \epsilon^{abcd}N_{bcd}\nabla_a \beta - \epsilon^{abcd}\delta N_{bcd}\nabla_a \Omega = 0\\
\Longleftrightarrow\quad & n u^a \nabla_a \beta - \left( u^a \delta n +n \delta u^a+\frac{1}{2} nu^a g^{bc}\delta g_{bc} \right)\nabla_a \Omega = 0 \\
\Longleftrightarrow\quad & u^a \nabla_a \beta = (\delta u^a)\nabla_a \Omega.
\ea
\ee
But we defined $\beta$ so as to satisfy the last equality, so we have shown that
\be 
\lie_t \delta N_{abc} = -\lie_\xi N_{abc} \, .
\ee
Thus, we have shown that $\lie_t \delta \mathcal Q$ can be represented in the Lagrangian displacement framework and has $\Delta j = 0$. 

Now, let $\eta^a$ be any axisymmetric trivial displacement. Then we have
\be
W[(0, \eta), \lie^2_t \delta \mathcal Q] = - W[(0, \lie_t \eta), \lie_t \delta \mathcal Q] = 0 \, ,
\ee
where the first equality follows from the same argument as used above to prove that $\mE$ is symmetric, and the second equality follows from the fact that $\lie_t \eta^a$ is an axisymmetric trivial displacement of the form $f \varphi^a$ (see subsection 
\ref{LagrangianformC}) and $\lie_t \delta \mathcal Q$ satisfies $\Delta j = 0$. Thus, $\lie^2_t \delta \mathcal Q$ is symplectically orthogonal to all trivial perturbations. Furthermore it follows immediately from conservation of ADM momentum that $\lie_t \delta \mathcal Q$ and $\lie^2_t \delta \mathcal Q$ have vanishing linearized momentum. Thus, $\lie^2_t \delta \mathcal Q \in \mathcal V$.
\end{proof}

Now, if $\delta \mathcal Q$ has exponential growth in time, then so does $\lie^2_t \delta \mathcal Q$. Therefore, the absence of exponentially growing solutions of the form $\lie^2_t \delta \mathcal Q$ implies the absence of any exponentially growing solutions at all. In view of this fact and the previous theorem, we have the following result:

\begin{theorem}
If $\mathcal E$ is non-negative on the subspace of axisymmetric perturbations in $\mathcal V$, then there are no smooth, axisymmetric solutions to the Einstein-fluid equations with suitable fall-off at infinity that have exponential growth in time, i.e., mode stability holds for all axisymmetric perturbations. Conversely, if ${\mathcal E} (\delta \phi, \delta \phi) < 0$ for some axisymmetric $\delta \phi \in \mathcal V$, then one has instability in the same sense as in Theorem \ref{dynstabthm}. \label{5.2}
\end{theorem}

\section{Thermodynamic Stability} \label{thermostabsec}

We turn our attention now to the thermodynamic stability of stars in thermal equilibrium. As explained in the Introduction, the criterion for thermodynamic stability is positivity of the quantity 
\be
{\mathcal E}' \equiv \delta^2 M  - \widetilde T \delta^2 S - \widetilde\mu \delta^2 N - \Omega \delta^2 J 
\ee
for all perturbations with $\delta M = \delta N = \delta J = 0$ (and, hence, $\delta S = 0$). In the case of dynamic stability, one can consider stability with respect to perturbations that lie in subspaces that are preserved under dynamic evolution, such as the subspace $\mV$ in Theorem \ref{dynstabthm}. However, the premise behind the notion of thermodynamic stability is that all states are accessible under the true dynamics, provided only that the fundamental conservation laws of $M$, $J$, and $N$ are respected. Thus, to prove thermodynamic stability of {\em any} perturbation, one must show positivity of $\mE'$ on {\em all} perturbations with $\delta M = \delta N = \delta J = 0$ (or, in the axisymmetric case, all axisymmetric perturbations with $\delta M = \delta N = \delta J = 0$).

We shall show that for all perturbations in the Lagrangian displacement framework, we have 
\be
\mE_r = \mE' \,,
\ee
where $\mE_r$ is the canonical energy in the ``rotating frame," i.e., defined with respect to the Killing field $v^a = t^a + \Omega \varphi^a$ to which $u^a$ is proportional. Thus, a necessary condition for thermodynamic stability is positivity of $\mE_r$ on all perturbations within the Lagrangian framework with $\delta J = 0$ (since $\delta N = \delta S = 0$ holds automatically for perturbations within the Lagrangian framework). Now, as previously mentioned, any perturbation with $\delta N = \delta S = 0$ can be described within the Lagrangian framework provided only that $\delta s/ |D_a s|$ is bounded \cite{Friedman}. Thus, for example, if the background star is such that $D_a s = 0$ at only one point and this zero is of order $1$, then the smooth Lagrangian perturbations are of co-dimension $1$ in the space of all smooth perturbations with $\delta N = \delta S = 0$. Thus, positivity of $\mE_r$ on Lagrangian perturbations with $\delta J = 0$ should also be ``nearly sufficient'' for thermodynamic stability. However, we shall not attempt to establish any sufficiency results along these lines, except for a remark about the isentropic case at the end of this section.

We begin by deriving an expression for the ordinary canonical energy, $\mE$, for perturbations of a star in dynamic (but not necessarily thermodynamic) equilibrium in terms of second order variations. Since $\mathcal E$ is gauge invariant, we may evaluate it in a gauge where the Lagrangian displacement vanishes. We thereby obtain
\be
\mathcal E\left(\delta_1\phi, \delta_2 \phi\right) = W^{(g)}\left[g_{ab}; \Delta_1 g_{ab}, \lie_t 
\Delta_2 g_{ab}\right]\,,
\ee
where $W^{(g)}$ is the ``gravitational part'' of $W$ (see \eqref{eq:symplecticform}), since the matter contribution, $W^{(m)}$, vanishes when the Lagrangian displacement vanishes. 
Now consider a 1-parameter family of solutions, $\phi (\lambda)$, corresponding to the perturbation $\delta \phi$, expressed in a gauge where the Lagrangian displacements vanish to all orders. Denoting the physical quantities in this gauge by $\widehat{\mathcal Q}(\lambda)$, we obtain
\be
\ba
\mathcal E (\delta\phi,\delta\phi) ={}& W^{(g)}\left[g_{ab}; \Delta g_{ab}, 
\lie_t \Delta g_{ab}\right] \\
={}& \left. W^{(g)}\left[\widehat g_{ab}(\lambda); \frac{d}{d\lambda} \widehat 
g_{ab}(\lambda), \lie_t \frac{d}{d\lambda} \widehat g_{ab}
(\lambda)\right]\right|_{\lambda=0} \\
={}& \left.\frac{d}{d\lambda}W^{(g)}\left[\widehat g_{ab}(\lambda); \frac{d}{d
\lambda} \widehat g_{ab}(\lambda), \lie_t \widehat g_{ab}(\lambda) 
\right]\right|_{\lambda=0} \\
={}& \left.\frac{d^2}{d\lambda^2}\widehat M(\lambda)\right|_{\lambda=0} \\
&+\int_{\Sigma} \left.\frac{d}{d\lambda}\left(t^a\left[ \frac{d}{d\lambda}\left(\widehat T_{a}^{\phantom ab}(\lambda)  \widehat \epsilon_{bdef}(\lambda)\right) -\frac{1}{2}  \widehat T^{bc}(\lambda) \frac{d\widehat g_{bc}(\lambda)}{d\lambda}\widehat \epsilon_{adef}(\lambda)  \right]\right)\right|_{\lambda=0} \\
={}& \delta^2 M +\int_{\Sigma} t^a \left[ \Delta^2\left( T_{a}^{\phantom ab}  \epsilon_{bdef}\right) -\frac{1}{2} \Delta \left(T^{bc}\Delta g_{bc} \epsilon_{adef} \right)   \right].
\ea
\label{canen2}
\ee
Here, in the third line, we have used the fact that $\widehat g(0) = g(0)$ is stationary. The fourth equality comes from the general identity \eqref{fundid} (and the field equations), and the definition of the ADM mass. The final equality uses the fact that the ADM mass is gauge invariant, so $\widehat M(\lambda) = M(\lambda)$. Below, we will simplify this expression further in the case where $\delta \phi$ is axisymmetric.

Now, consider the case where the background star is in thermodynamic equilibrium, so that, in particular, $\Omega$ is constant. Then the background $4$-velocity $u^a$ is proportional to the Killing field
\be
v^a = t^a + \Omega \varphi^a
\ee
and it is natural to consider the quantity
\be
\mathcal E_r (\delta_1 \phi, \delta_2 \phi) = W[\phi; \delta_1 \phi, \lie_v \delta_2 \phi] \, ,
\label{canenrot}
\ee
which may be interpreted as the ``canonical energy as measured in the frame that rotates rigidly with the star.'' A calculation parallel to the above calculation yields
\be\label{canenrot2}
\mathcal E_r (\delta\phi,\delta\phi) = \delta^2 M - \Omega \delta^2 J +\int_{\Sigma} v^a \left[ \Delta^2\left( T_{a}^{\phantom ab}  \epsilon_{bdef}\right) -\frac{1}{2} \Delta \left(T^{bc}\Delta g_{bc} \epsilon_{adef} \right)   \right],
\ee
where $J$ denotes the ADM angular momentum and the presence of the additional term $-\Omega \delta^2 J$ arises simply because of the asymptotic behavior of $v^a$ (as compared with $t^a$) at infinity. However, the last term can be seen to vanish using Lemma \ref{fluididentity}. Namely, the right hand side of the identity \eqref{fluidid} is zero in a gauge where the Lagrangian displacements vanish (since $\bN$ and $\bS$ are then fixed) so the identity evaluated at $\lambda=0$ in such a gauge says
\be\label{identityinDeltagauge}
0=u^a \left[ \frac{1}{2} T^{bc}\Delta g_{bc} \epsilon_{adef} - \Delta\left(T_a^{\phantom ab} \epsilon_{bdef}\right) \right],
\ee
and the $\lambda$ derivative of the identity evaluated at $\lambda=0$ in such a gauge yields
\be \label{secondorderDelta}
\ba
0&=\Delta\left(u^a \left[ \frac{1}{2} T^{bc}\Delta g_{bc} \epsilon_{adef} - \Delta\left(T_a^{\phantom ab} \epsilon_{bdef}\right) \right]\right)\\
&=u^a\Delta\left[ \frac{1}{2} T^{bc}\Delta g_{bc} \epsilon_{adef} - \Delta\left(T_a^{\phantom ab} \epsilon_{bdef}\right) \right]+\left(\Delta u^a\right) \left[ \frac{1}{2} T^{bc}\Delta g_{bc} \epsilon_{adef} - \Delta\left(T_a^{\phantom ab} \epsilon_{bdef}\right) \right] 
\\
&=u^a\left[ \frac{1}{2} \Delta\left(T^{bc}\Delta g_{bc} \epsilon_{adef} \right)- \Delta^2\left(T_a^{\phantom ab} \epsilon_{bdef}\right) \right],
\ea
\ee
where we have used the fact that $\Delta u^a$ is parallel to $u^a$ (equation \eqref{Deltau}) and \eqref{identityinDeltagauge}. Thus, the last term in \eqref{canenrot2} vanishes, and we obtain
\be
\mathcal E_r (\delta\phi,\delta\phi) =  \delta^2 M - \Omega \delta^2 J .
\ee
Taking account of the fact that we automatically have $\delta^2 N = \delta^2 S = 0$ for all variations describable within the Lagrangian framework, we see that, for perturbations within the Lagrangian framework, the quantity ${\mathcal E}_r$ coincides with the quantity $\mathcal E'$, as we desired to show. Consequently, we immediately obtain the following theorem;

\begin{theorem}
For a star in thermodynamic equilibrium, a necessary condition for thermodynamic stability is positivity of $\mathcal E_r$ on all linearized solutions within the Lagrangian framework that have $\delta J = 0$. \label{6.1}
\end{theorem}

The identification of $\mathcal E_r$ as the quantity whose positivity determines thermodynamic stability is in accord with the analysis of Lindblom and Hiscock \cite{LindblomHiscock}. Lindblom and Hiscock further argued that dissipative processes will act to {\em stabilize} a star against the CFS dynamic instability implied by theorem \ref{dynstabthm}, which might seem to suggest that a star could be dynamically unstable (to the CFS instability) but thermodynamically stable. However, by the general arguments given in the Introduction, this is impossible. Indeed, it is easy to see that all rotating stars are thermodynamically unstable: A perturbation that corresponds to a gravitational wavepacket localized far from the star and in a high angular momentum state---with negligible perturbation to the star itself---can easily be made\footnote{\footnotesize This can be done by choosing the perturbation to be predominantly composed of modes of frequency $\omega$ and angular quantum number $m$ such that $0 < \omega < m \Omega$. \medskip} to have $\Omega \delta^2 J > \delta^2 M$, and, hence, $\mathcal E_r < 0$. In other words, it is always entropically favorable to put some of the angular momentum of the star in low energy gravitational radiation, and then use the rotational energy thereby gained to add thermal energy to the star. If gravitational radiation were assigned a non-zero entropy, it would be even more entropically favorable to do this. 

The resolution of the apparent discrepancy between this argument and the results of Lindblom and Hiscock is that Lindblom and Hiscock restricted consideration only to ``short length scale perturbations'' that are localized within the star and have negligible metric perturbation, thereby excluding the entropically favorable perturbations of the previous paragraph. For these short length scale perturbations, the condition for thermodynamic stability reduces to the two relations \eqref{localthermostab} for the thermodynamic stability of a homogeneous system. In other words, a star will have positive $\mE_r$ for short length scale perturbations if and only if \eqref{localthermostab} holds at each point in the star. Note that the first of these relations is equivalent to the Schwarzschild stability criterion obtained by Lindblom and Hiscock; the second condition does not appear in Lindblom and Hiscock's analysis, presumably because they did not consider the dissipative process of diffusion. In any case, since the CFS dynamic instability is of this short length scale type, it is highly plausible that, if the local thermodynamic stability criteria \eqref{localthermostab} hold, then dissipative processes will damp this instability, as claimed by Lindblom and Hiscock. Nevertheless, although dissipative processes may enormously increase the timescale---beyond the already possibly enormous timescale of the CFS instability---they cannot prevent the star from eventually reaching a state of higher entropy by radiating its angular momentum away into modes with $0 < \omega < m \Omega$.

Finally, we return to the expression \eqref{canen2} for the canonical energy for a background star that is in dynamic---but not necessarily thermodynamic---equilibrium. If we restrict consideration to axisymmetric perturbations (so that, in particular, $\Delta^k \varphi^a = 0$), then we have---making use of \eqref{secondorderDelta} and again choosing $\Sigma$ to be axisymmetric---
\be
\ba
&\int_{\Sigma} t^a \left[ \Delta^2\left( T_{a}^{\phantom ab}  \epsilon_{bdef}\right) -\frac{1}{2} \Delta \left(T^{bc}\Delta g_{bc} \epsilon_{adef} \right)   \right] \\
=&\int_{\Sigma} \left( |v|u^a-\Omega\varphi^a\right) \left[ \Delta^2\left( T_{a}^{\phantom ab}  \epsilon_{bdef}\right) -\frac{1}{2} \Delta \left(T^{bc}\Delta g_{bc} \epsilon_{adef} \right)   \right] \\
=& - \int_\Sigma\Omega \Delta^2\left( \varphi^a T_{a}^{\phantom ab}  \epsilon_{bdef}\right) \\
=& - \int_\Sigma\Omega\Delta^2J_{def} 
\ea
\ee
so we obtain
\be
\mathcal E (\delta\phi,\delta\phi) = \delta^2M - \int_\Sigma\Omega\Delta^2\bJ \, ,
\ee
which is our desired general expression for $\mathcal E$ in terms of second order variations for a star in dynamic equilibrium.

If the star is rigidly rotating, the last term becomes 
\be
\int_\Sigma\Omega\Delta^2 \bJ  = \Omega \int_\Sigma \Delta^2 \bJ = \Omega \int_\Sigma \delta^2 \bJ = \Omega \delta^2 J \, ,
\ee
where the last equality follows from an argument similar to the argument that led to \eqref{Jeq2}. 
Thus, for an axisymmetric perturbation of a star in thermodynamic equilibrium, we have\footnote{\footnotesize This result can also be seen directly from the fact that $\mathcal E_r (\delta\phi,\delta\phi) - \mathcal E (\delta\phi,\delta\phi) = \Omega W[\phi, \lie_\varphi \phi] = 0$.\medskip}
\be
\mathcal E (\delta\phi,\delta\phi) = \mathcal E_r (\delta\phi,\delta\phi) \, .
\ee
As an immediate consequence, we have the following theorem:

\begin{theorem}For a star in thermodynamic equilibrium, a necessary condition for thermodynamic stability with respect to axisymmetric perturbations is positivity of $\mE$ on all axisymmetric linearized solutions within the Lagrangian framework that have $\delta J = 0$. \label{6.2}
\end{theorem}

As a simple application of our results, consider a star at $T=0$ for which the entropy per particle, $s$, takes its minimum value $s=0$ throughout the star\footnote{We assume here that when $T=0$, the condition for thermodynamic stability remains positivity of $\mathcal E'$ for perturbations with $\delta N = \delta S = \delta J = 0$, even though the argument given for this criterion in the Introduction assumed $T > 0$.}. Then any perturbation for which $\delta S = 0$ must have $\delta s = 0$ everywhere. Hence, for this ``isentropic case,'' every perturbation with $\delta S = \delta N = 0$ can be described in the Lagrangian framework. Consequently, in this case, the word ``necessary'' can be replaced by ``necessary and sufficient'' in Theorems \ref{6.1} and \ref{6.2}. Now consider spherically symmetric perturbations of a static, spherically symmetric isentropic star. Such perturbations obviously have $\delta J = 0$. It is not difficult to show that there do not exist any spherically symmetric trivial perturbations. Consequently, we have $\mV = \mathcal C$. Comparison of Theorems \ref{5.2} and \ref{6.2} (with ``necessary'' replaced by ``necessary and sufficient'') then immediately yields the following result\footnote{\footnotesize It should be noted that the relationship between dynamic and thermodynamic stability has been studied extensively in the special case of spherically symmetric perturbations of static, spherically symmetric stars. In particular, Sorkin, Wald, and Zhang \cite{SWZ} showed equivalence between dynamic and thermodynamic stability in the case of a radiation fluid, while Roupas \cite{Roupas} showed equivalence for any fluid with zero chemical potential (which includes the radiation case). In addition, Gao \cite{Gao1, Gao2} showed for general fluids that if a spherically symmetric configuration is at an extremum of total entropy with respect to spherically symmetric perturbations that fix total mass and particle number, then the configuration is in static equilibrium. \medskip}: {\em In the isentropic case, for spherically symmetric perturbations of static, spherically symmetric stars, thermodynamic stability is equivalent to dynamic stability.} 

\begin{acknowledgments}

  We wish to thank John Friedman for many helpful discussions. This
  research was supported in part by NSF grant PHY~12-02718 to the
  University of Chicago, and by NSERC\@.  S.~R.~G. is supported by a
  CITA National Fellowship at the University of Guelph, and he thanks
  the Perimeter Institute for hospitality.

\end{acknowledgments}

\appendix

\section{Existence of Desired Solutions to the Linearized Constraints} \label{constraintsproof}

Let $\Sigma$ be a $t-\varphi$ reflection invariant Cauchy surface for a star in dynamic equilibrium. Let $\boldsymbol e$ be a fixed, non-dynamical volume element on $\Sigma$, so the volume element associated with the induced metric on $\Sigma$ is $\sqrt{h} \boldsymbol e$. Let $\nu^a$ be the future-directed unit normal to $\Sigma$. Consider perturbations off of this background. The linearized Hamiltonian constraint on $\Sigma$ is 
\be
0 = \delta \left(\nu^a C_a \right) =  \left[-\frac{1}{8 \pi} \delta \left(\sqrt{h} \nu^a\nu^bG_{ab} \right) +  \delta \left(\sqrt{h} \nu^a\nu^bT_{ab}\right) \right]
\ee
and the linearized momentum constraint is
\be
0 = \delta \left(h_a^{\phantom ab} C_b \right) = \left[-\frac{1}{8 \pi} \delta \left(\sqrt{h}  h_a^{\phantom ab} \nu^cG_{bc} \right) +  \delta \left(\sqrt{h} h_a^{\phantom ab} \nu^cT_{bc}\right) \right] \, .
\ee
All quantities appearing in these equations can be expressed\footnote{\footnotesize Note, in particular, that $\delta (u^a \nu_a)$ can be expressed in terms of $\delta(h_{ij} u^i u^j)$ and background quantities  on account of the normalization condition on $u^a$.} in terms of background quantities and the perturbation quantities $\delta h_{ij}$, $\delta \pi^{ij}$, $\delta n$, $\delta s$, and $\delta u^i$, where $u^i$ denotes the $3$-velocity, i.e., the projection of $u^a$ tangent to $\Sigma$. Alternatively, we can replace the fluid quantities $(\delta n, \delta s, \delta u^i)$ with $(\delta {\mathcal N}, \delta {\mathcal S}, \delta {\mathcal J}, \delta u^i_\perp)$, where $\delta u^i_\perp$ is the projection of $\delta u^i$ perpendicular to $\varphi^i$ and 
\be \label{densitydefs}
\ba
\mathcal N &\equiv -\sqrt{h\,} n (u^a \nu_a)\,, \\
\mathcal S &\equiv s\, \mathcal N\,, \\
\mathcal J &\equiv j\,  \mathcal N = \frac{\rho+p}{n} (u^a \varphi_a) \mathcal N, \\
\ea
\ee
so that 
\be
\ba
\overline{\delta \bN} &= (\delta {\mathcal N}) {\boldsymbol e}\,, \\
\overline{\delta \bS} &= (\delta {\mathcal S}) {\boldsymbol e}\,, \\
\overline{\delta \bJ} &= (\delta {\mathcal J}) {\boldsymbol e}\,, \\
\ea
\ee
where $\overline{\delta \bN}$, $\overline{\delta \bS}$, and $\overline{\delta \bJ}$ denote the pullbacks of $\delta \bN$, $\delta \bS$, and $\delta \bJ$ to $\Sigma$.

In terms of these variables, the linearized Hamiltonian constraint takes the form
\be
\ba
\frac{1}{16\pi} \Biggl[ &-R^{ij}(h) \delta h_{ij} + D^i D^j \delta h_{ij} - D^i D_i \delta h_j^{\phantom jj} + h^{-1} \pi^{ij} \pi_{ij} \delta h_k^{\phantom kk}-2h^{-1} \pi_{ij}\delta\pi^{ij}\\
&  - 2 h^{-1} \pi_i^{\phantom ij}\pi^{ik} \delta h_{jk} +\sqrt{h}\left(\nu^a\nu^b G_{ab}\right)\delta h_j^{\phantom jj} +8\pi\sqrt{h}T^{ij}\delta h_{ij}\Biggr] \\ &\qquad\qquad\qquad\qquad\qquad\qquad\quad =  -\frac{1}{u^a \nu_a} \left[\mu \delta\mathcal N + T \delta\mathcal S + \frac{u^a \varphi_a}{\varphi^b \varphi_b} \delta\mathcal J \right]   \, ,
\label{lhc}
\ea
\ee
where we have used the fact that since the background spacetime is $t - \varphi$ symmetric about $\Sigma$, the background $\pi^{ij}$ must be
odd under the action of the reflection isometry $\varphi \to - \varphi$ of $h_{ij}$ on $\Sigma$, so, in particular
${\pi^i}_i = 0$. The $\varphi^i$-component of the linearized momentum constraint is 
\be
\ba
\frac{1}{16\pi} \varphi^i \Biggl[ 2\sqrt{h\,}& D_j \left( h^{-1/2} \delta \pi_i^{\phantom ij} \right) + 2\pi^{jk} D_j \delta h_{ik}\\
&  - \pi^{jk} D_i \delta h_{jk} + 2\sqrt{h\,}\delta h_{ik}D_j \left( h^{-1/2} \pi^{jk}\right) \Biggr] = -\delta \mathcal J \, ,
\label{lmcpa}
\ea
\ee
and the components of the linearized momentum constraints perpendicular to $\varphi^a$ are
\be
\ba
\frac{1}{16\pi} \Biggl[ 2\sqrt{h\,}& D_j \left( h^{-1/2} \delta \pi^{ij} \right) + 2\pi^{jk} D_j \delta h^i_{\phantom ik}- \pi^{jk} D^i \delta h_{jk}\Biggr]_\perp
=  \sqrt{h\,}(u^b \nu_b)(\rho+p) \delta u^i_{\perp} \, ,
\label{lmcpe}
\ea
\ee
where the subscript ``$\perp$'' means the projection orthogonal to $\varphi^i$ in $\Sigma$. Note that $\delta u^i_{\perp}$ does not appear at all in eqs. \eqref{lhc} and \eqref{lmcpa}.

The following lemma is needed in Theorem \ref{fluididentity} to show that we can solve the constraints for any choice of axisymmetric $\overline{\delta \bN}$, $\overline{\delta \bS}$, and $\overline{\delta \bJ}$:

\medskip

\noindent
{\bf Lemma A.1:} {\em Let $\delta {\mathcal N}$, $\delta {\mathcal S}$, and $\delta {\mathcal J}$ be specified arbitrarily as smooth, axisymmetric functions with support inside the background star, such that $\delta {\mathcal J}/\varphi^a \varphi_a$ also is smooth (i.e., $\delta \mathcal J$ vanishes on the ``rotation axis''). Then we can choose the remaining initial data $(\delta h_{ij}, \delta \pi^{ij}, \delta u^i_\perp)$ so as to solve the linearized constraints \eqref{lhc}--\eqref{lmcpe}.}

\smallskip
\noindent
{\bf Proof:} We choose $\delta h_{ij}$ and $\delta \pi^{ij}$ to be of the form
\be \label{ansatz}
\ba
\delta h_{ij} &=  \psi h_{ij} \,,\\ 
\delta \pi^{ij} &= \sqrt{h\,} D^{(i} F^{j)} - \psi \pi^{ij} \, .
\ea
\ee
We choose $F^i$ to satisfy
\be
D^j \left(  D_{(i} F_{j)} \right) = - 8 \pi \frac{\delta \mathcal J}{\varphi^a \varphi_a} \varphi_i
\label{Feq}
\ee
Since the right side is a smooth vector field of compact support, by standard arguments (see, e.g., \cite{ChruscielDelay}), there exists a unique solution to \eqref{Feq} that goes to zero at infinity. Since the right side is axisymmetric and is odd under the action of the reflection isometry $\varphi \to - \varphi$ of $h_{ij}$ on $\Sigma$, the same must be true of $F^i$. It may be straightforwardly verified that for $\delta h_{ij}$ and $\delta \pi^{ij}$ of the form \eqref{ansatz} together with this choice of $F^i$ and with {\em any} choice of $\psi$, the $\varphi$-component \eqref{lmcpa} of the linearized momentum constraint is satisfied. 

Now substitute \eqref{ansatz} into the linearized Hamiltonian constraint \eqref{lhc}. We obtain
\be \label{Hconstraint}
-D_i D^i \psi + M \psi = h^{-1/2} \pi^{ij} D_iF_j + 
\frac{8 \pi  h^{-1/2} }{u^a \nu_a} \left[\mu \delta\mathcal N + T \delta\mathcal S + \frac{u^a \varphi_a}{\varphi^b \varphi_b} \delta\mathcal J \right] \, ,
\ee
where
\be
M \equiv  h^{-1} \pi^{ij} \pi_{ij} + 8\pi \frac{(\varphi^a u_a)^2}{\varphi^b\varphi_b} (\rho + p)  + 4\pi (\rho+3p).
\ee
Since $M$ is manifestly non-negative, and since the right side of \eqref{Hconstraint} vanishes suitably rapidly at infinity, by standard arguments \cite{CantorBrill, ChruscielDelay}, there exists a unique solution, $\psi$, of this equation that vanishes at infinity. With our previous choice of $F^i$ and this choice of $\psi$, our ansatz \eqref{ansatz} solves both \eqref{lhc} and \eqref{lmcpa}.

Finally, we note that with our choice $\delta h_{ij}$ and $\delta \pi^{ij}$, the vector inside the square brackets on the left side of \eqref{lmcpe} is axisymmetric and is odd under the action of the reflection isometry $\varphi \to - \varphi$ on $\Sigma$. It follows that the projection of this vector perpendicular to $\varphi^i$ vanishes. Consequently, we may solve the remaining constraint \eqref{lmcpe} by choosing $\delta u^i_\perp = 0$. $\Box$

\section{Phase Space Construction} \label{phasespaceproof}

In this Appendix, we will show that $\delta\phi$ is a degeneracy of $W$ (given by \eqref{eq:symplecticform}) if and only if $0=\delta h_{ij}=\delta \pi^{ij}_{\phantom{ij}klm} = q^a_{\phantom a b}\xi^b = \delta u^i$ on $\Sigma$. This shows that phase space may be identified with the set of fields $(h_{ij},\boldsymbol\pi^{ij},\psi, u^i)$ on a Cauchy surface $\Sigma$, where $\psi$ is a diffeomorphism from the space of fiducial flowlines to $\Sigma$. We will then obtain canonical coordinates on phase space.

It is clear from \eqref{eq:symplecticform} that $W$ depends at most on the following quantities on $\Sigma$ (for each of the two perturbations): $\delta h_{ij}$, $\delta \pi^{ij}_{\phantom{ij}klm}$, $\delta N$ (the perturbed lapse), $\delta N_i$ (the perturbed shift), $\xi^a$, and the normal derivative of $\xi^a$. Using \eqref{lagrangianpert}, we may write
\be
\ba
W[\delta_1 \phi,\delta_2 \phi] &= W[\delta_1\phi,\Delta_2 \phi] - W[\delta_1 \phi,\mathcal L_{\xi_2} \phi ], \\
\ea
\ee
where $\Delta_2 \phi$ denotes the perturbation $(\delta g_{ab} = \Delta_2 g_{ab},\xi^a = 0)$. We use \eqref{eq:symplecticform} to evaluate the first term and \eqref{ham0} to evaluate the second, thereby obtaining
\be
\ba
W[\delta_1 \phi,\delta_2 \phi] &= \int_\Sigma \biggl[\frac{1}{16\pi}\left(\Delta_2 h_{ij} \delta_1 \boldsymbol\pi^{ij}  -\delta_1 h_{ij} \Delta_2 \boldsymbol\pi^{ij} \right)\\
&\qquad \qquad- \xi_1^a\Delta_2  \boldsymbol P_{a} -i_{\xi_2}(\bE \cdot \delta_1\phi)-\delta_1 \boldsymbol{C}_{\xi_2} \biggr],
\ea
\ee
where $\Delta_2 h_{ij}$ denotes the perturbed spatial metric, $\delta h_{ij}$, associated with the perturbation $\Delta_2 \phi$---i.e., it does {\em not} mean $\delta_2 h_{ij}+\lie_{\xi_2} h_{ij}$---and similarly for $\Delta_2 \pi^{ij}_{\phantom{ij}klm}$, $\Delta_2 N$, and $\Delta_2 N_a$. We write
\be \label{Adefn}
\Delta \boldsymbol P_a = \boldsymbol A_a^{\phantom abc} \Delta g_{bc},
\ee
where
\be
  A_{a\phantom{bc}def}^{\phantom abc} = -\frac{1}{2}  u^b u^c P_{adef}- \frac{1}{2} c_s^2  q^{bc}P_{adef} + \left(\frac{\rho+p}{n}\right)q_a^{\phantom a (b}u^{c)}N_{def}.
\ee
We note that $u^a \boldsymbol A_a^{\phantom abc}=0$, since $u^a \boldsymbol P_a=0$ (see \eqref{Pdefn}). Writing $\Delta g_{ab}$ in terms of the perturbations to the lapse, shift, and spatial metric, we obtain
\be \label{Wfordegeneracy}
\ba
W[\delta_1 \phi,\delta_2 \phi] = \int_\Sigma \Biggl[ \frac{1}{16\pi}&\left(\Delta_2 h_{ij} \delta_1 \boldsymbol{\pi}^{ij} - \delta_1 h_{ij} \Delta_2 \boldsymbol{\pi}^{ij}\right) -i_{\xi_2}(\bE \cdot \delta_1\phi)-\delta_1 \boldsymbol{C}_{\xi_2} \\
& -  \xi_1^a \boldsymbol A_a^{\phantom abc} \left( \Delta_2 h_{bc} -\frac{2}{N} \nu_b \nu_c \Delta_2 N - \frac{2}{N} \nu_b \Delta_2 N_c \right)  \Biggr],
\ea
\ee
where $\nu^a$ is the unit normal to $\Sigma$. The quantities $\Delta_2 h_{ij}$, $\Delta_2 \boldsymbol\pi^{ij}$, $\Delta_2 N$, $\Delta_2 N_i$, and $\xi_2^a$ on $\Sigma$ can be varied independently. (Note that the field equations and the linearized constraints are not being imposed here, since we are seeking the degeneracy directions of $W$ in the full field space, not merely in the solution space.) Thus, $\delta \phi \equiv \delta_1 \phi$ is a degeneracy of $W$ if and only if the coefficients of $\Delta_2 h_{ij}$, $\Delta_2 \boldsymbol\pi^{ij}$, $\Delta_2 N$, $\Delta_2 N_i$, and $\xi_2^a$ in (the pullback to $\Sigma$ of) the integrand of \eqref{Wfordegeneracy} are each individually zero, i.e., if and only if the following conditions hold:
\begin{align}
\label{dc1} 0&= \frac{1}{16 \pi}\delta \pi^{ij}_{\phantom{ij}klm}+\xi^a A_{a\phantom{ij}klm}^{\phantom aij}\,,\\
\label{dc2} 0&= \delta h_{ij} \,,\\
\label{dc3} 0&=\xi^a A_{a\phantom{bc}klm}^{\phantom abc}\nu_b \nu_c\,,\\
\label{dc4} 0&=\xi^a A_{a\phantom{bi}klm}^{\phantom abi}\nu_b\,,\\
\label{dc5} 0&=\delta  \overline{\boldsymbol C}_{a}+ \overline{i_\nu ( \boldsymbol E \cdot \delta\phi)}\nu_{a}\,.
\end{align}
Conditions \eqref{dc3} and \eqref{dc4} together imply
\be
0 = \xi^a A_{a\phantom{bc}klm}^{\phantom abc}\nu_b = \frac{1}{2}(\rho+p)\left[(u^b \nu_b)^2 \delta^c_{\phantom cd} - c_s^2 \nu_b \nu_d q^{bc} \right] \xi^a q_a^{\phantom a d} \nu^e \epsilon_{eklm} \, .
\label{xiperp}
\ee
Using $c_s^2 \leq 1$ (see \eqref{Tpc}), we see that the right side cannot vanish unless
\be
\xi^a q_a^{\phantom a d} = 0\,,
\ee
i.e., $\xi^a$ is proportional to $u^a$. It also follows from \eqref{dc1} and \eqref{xiperp} that $\delta \boldsymbol\pi^{ij}= 0$.

Thus, we have shown that \eqref{dc1}--\eqref{dc4} are equivalent to $\delta h_{ij} = 0$, $\delta \boldsymbol\pi^{ij} = 0$, and $\xi^a \propto u^a$. We now show that, in the presence of these conditions, the final condition \eqref{dc5} is equivalent to $\delta u^i=0$. Let $\delta_1 \phi$ be such that $0=\delta_1 h_{ab}=\delta_1 \boldsymbol\pi^{ab}=q^a_{\phantom ab}\xi_1^b$ on $\Sigma$. We will show that $\delta_1 \phi$ is a degeneracy of $W$ if and only if $\delta_1 u^i = 0$. 

To show this, we write
\be
\xi_1^a = f u^a + \tau \zeta^a,
\ee
where $\tau$ is a smooth function that vanishes on $\Sigma$ and is such that $\nabla_a \tau = \nu_a$ on $\Sigma$ and $\zeta^a u_a = 0$ everywhere. 
Then we have
\be
W[\delta_1 \phi, \delta_2 \phi] = W[(0, f u^a), \delta_2 \phi]+W[(\delta_1 g_{ab} , \tau \zeta^a), \delta_2 \phi] = W[(\delta_1 g_{ab} , \tau \zeta^{a}), \delta_2 \phi],
\ee
since the flowline trivial $(0, f u^a)$ is automatically a degeneracy of $W$ since it satisfies\footnote{\footnotesize Any trivial satisfies $\delta \bC_a = 0$, and a flowline trivial satisfies $E\cdot\delta\phi=-fu_b \nabla_a T^{ab}=0$ since the $u^a$-component of stress energy conservation is automatically satisfied in the Lagrangian formalism.\medskip} \eqref{dc5}. Using \eqref{eq:symplecticform}, we obtain
\be
\ba
W[\delta_1 \phi, \delta_2 \phi] &= \int_\Sigma \left( \xi_2^a \delta_1 P_{apqr} - [\tau \zeta, \xi_2]^a P_{apqr} \right)\\
&= \int_\Sigma \left( \xi_2^a \delta_1 P_{apqr} + \xi_2^b \nu_b\zeta^a P_{apqr} \right) \, .
\ea
\ee
We now eliminate $\zeta^a$ in terms of $\delta_1 g_{ab}$ and $\delta_1 (h^a_{\phantom ab}u^b)$. 
A lengthy calculation yields
\be
W[\delta_1 \phi, \delta_2 \phi] = \int_{\Sigma} \left(\rho+p\right)\xi^a_2 B_{ab} \delta_1 \left( h^{bc} u_c \right) \, ,
\label{wdeg}
\ee
where
\be
B_{ab} = -2 u^c h_{b[a}\nu_{c]} + \frac{c_s^2}{(u^d \nu_d)^2}q_{ac}\nu^c u_b \, .
\ee
Again, using $c_s^2 \leq 1$, we find that the right side of \eqref{wdeg} vanishes for all $\xi_2^a$ if and only if $\delta_1 ( h^{bc} u_c) = 0$, as we desired to show.

Thus, we have shown that $\delta \phi$ is a degeneracy of $W$ if and only if the quantities $(\delta h_{ij}, \delta \boldsymbol \pi^{ij}, q^a_{\phantom ab}\xi^b, \delta u^i)$ vanish on $\Sigma$. These quantities are the first order variations of the quantities
$(h_{ij}, \boldsymbol \pi^{ij}, \psi, u^i)$ on $\Sigma$, where $\psi$ is a diffeomorphism from the space of fiducial flowlines, $\Sigma'$, to $\Sigma$ (see subsection \ref{LagrangianformB}). Thus, phase space is described by the quantities $(h_{ij}, \boldsymbol \pi^{ij}, \psi, u^i)$ on, $\Sigma$.

The variables $(\psi,u^i)$ are not canonically conjugate, as the symplectic product of two pure $\psi$ perturbations (keeping $u^i$ fixed) is not necessarily zero. One can obtain canonically conjugate variables by representing the dynamical diffeomorphism $\chi$ as a set of four ``coordinate'' scalar fields, following, e.g., \cite{LeeWald, Iyer}. Let $x^{\prime \mu'}$ be coordinates on $M'$ such that $x^{\prime1}$, $x^{\prime2}$, and $x^{\prime3}$ are constant along the fiducial flowlines. Then we can encode the information in $\chi$ via the 4 scalar fields
\be
x^{\mu'}\equiv x^{\prime\mu'} \circ \chi^{-1}.
\ee
The diffeomorphism $\psi:\Sigma\rightarrow\Sigma'$ is specified by giving $\left.x^{ i}\right|_\Sigma$ for $i=1,2,3$. 
Using the Lagrangian \eqref{L}, and following the prescription of section \ref{diffco}, one then finds that the matter part of the symplectic form is
\be
W^{(m)}[\phi;\delta_1\phi,\delta_2\phi] = \int_\Sigma \sum_{\mu=0}^{3}\left(\delta_2 x^{\mu} \delta_1 \boldsymbol p_{{\mu}} - \delta_1 x^{\mu} \delta_2 \boldsymbol p_{{\mu}} \right),
\ee
where
\be
\boldsymbol p_{{\mu}} \equiv (\chi^{-1}_\ast)_{\mu}^{\phantom{\mu} a} \boldsymbol P_{a},
\ee
with  $(\chi^{-1}_\ast)_{\mu}^{\phantom{\mu} a}$ being the inverse of $(\chi_\ast)^{\mu}_{\phantom{\mu} a} = \nabla_a x^{\mu} $. 
Since $(\chi^{-1}_\ast)_{\mu}^{\phantom{\mu} a} \nabla_a x^\nu = \delta_\mu^{\phantom\mu\nu}$, it follows that
\be
(\chi^{-1}_\ast)_{0}^{\phantom{0} a} \propto u^a,
\ee
since both sides annihilate $\nabla_a x^i$. Consequently,
\be
p_{0} \propto  u^a \boldsymbol P_a = 0.
\ee
Thus, we obtain
\be
W^{(m)}[\phi;\delta_1\phi,\delta_2\phi] = \int_\Sigma \sum_{i=1}^{3}\left(\delta_2 x^{i} \delta_1 \boldsymbol p_{{i}} - \delta_1 x^{i} \delta_2 \boldsymbol p_{{i}} \right).
\ee
The variables $x^{i}$ and $\boldsymbol p_{i}$ for $i=1,2,3$ are thus canonically conjugate.

\bibliography{mybib}

\end{document}